\documentclass[a4paper,reqno,11pt]{amsart}

\usepackage[utf8]{inputenc}
\usepackage[T1]{fontenc}
\usepackage{lmodern}
\usepackage[UKenglish]{babel}
\usepackage[UKenglish]{isodate}
\usepackage{amsmath, amsfonts, amssymb, amsthm, mathtools, cancel,bm} 
\usepackage{enumerate}
\usepackage{stmaryrd}
\usepackage{empheq}
\usepackage{geometry} 
\usepackage{mathrsfs} 
\usepackage{hyperref} 
\hypersetup{colorlinks = true}
\hypersetup{linktocpage} 
\usepackage[capitalise]{cleveref} 
\usepackage{dsfont}
\usepackage{calligra}
\usepackage{braket}

\theoremstyle{plain}
\newtheorem{thm}{Theorem}[section]

\newtheorem{lem}[thm]{Lemma}

\theoremstyle{definition}

\theoremstyle{remark} 
\newtheorem{rmkx}[thm]{Remark}

\newenvironment{rmk}
  {%
   \pushQED{\qed}\begin{rmkx}}
  {\popQED\end{rmkx}}

\usepackage{dsfont}
\usepackage{geometry}
\newcommand{\interval}[4]{\mathopen{#1}#2
	\mathclose{}\mathpunct{},#3
	\mathclose{#4}}

\newcommand{\intoc}[2]{\interval{(}{#1}{#2}{]}}
\newcommand{\intco}[2]{\interval{[}{#1}{#2}{)}}
\newcommand{\intoo}[2]{\interval{(}{#1}{#2}{)}}
\newcommand{\intinteger}[2]{\interval\llbracket{#1}{#2}\rrbracket}

\newcommand{\smallo}[1]{o\mathopen{}\left(#1\right)}
\newcommand{\bigO}[1]{O\mathopen{}\left(#1\right)}

\newcommand{\abs}[1]{\left\lvert#1\right\rvert}
\newcommand{\norm}[1]{\left\lVert#1\right\rVert}

\newcommand{\setst}[2]{\left\{#1\mathrel{}\middle|\mathrel{}#2\right\}}

\newcommand{\et}{\quad \text{and} \quad}

\newcommand{\hence}{\quad \text{hence} \quad}

\newcommand{\numberset}[1]{\mathbb{#1}}

\newcommand{\Z}{\numberset{Z}}

\newcommand{\R}{\numberset{R}}
\newcommand{\C}{\numberset{C}}

\newcommand{\diff}{\mathop{}\mathopen{}\mathrm{d}}
\DeclareMathOperator{\dd}{div}
\DeclareMathOperator{\curl}{curl}

\newcommand{\schatten}[1]{\mathfrak{S}_{#1}}
\newcommand{\FF}{\mathcal{F}_{\mathrm{PV}}} 
\newcommand{\hd}{\dot{H}^1_{\mathrm{div}}(\R^3)}
\newcommand{\Cgaugesobolev}{\dot{H}^1_\mathrm{div}\left(\R^3\right)}

\DeclareMathOperator{\tr}{tr}

\title[The PV-Regularised Dirac Field in Electromagnetic Fields at Thermal Equilibrium]{The Pauli-Villars-Regularised Dirac Field Coupled to Electromagnetic Fields at Thermal Equilibrium}
\author[William Borrelli]{William Borrelli}
\address{Politecnico di Milano, Dipartimento di Matematica, P.zza Leonardo da Vinci, 32, 20133 Milano, Italy}
\email{\href{mailto:william.borrelli@polimi.it}
{\nolinkurl{william.borrelli@polimi.it}}}
\author[Umberto Morellini]{Umberto Morellini}
\address{Università di Pisa, Dipartimento di Matematica, Largo Bruno Pontecorvo 5, 56127 Pisa, Italy}
\email{\href{mailto:umberto.morellini@dm.unipi.it}{\nolinkurl{umberto.morellini@dm.unipi.it}}}
\date{13 July 2026}

\begin{document}

\begin{abstract}
   In this paper we consider a model for the Dirac field in a classical electromagnetic field at thermal equilibrium. We adopt the Pauli-Villars regularisation technique in order to properly define the free energy of the equilibrium state, extending the previous work by the second named author on the purely magnetic case. This work provides a first step towards understanding polarisation effects in the Dirac field at thermal equilibrium, in the presence of both electrostatic and magnetic potentials. 
\end{abstract}

\maketitle

\tableofcontents

\addtocontents{toc}{\protect\setcounter{tocdepth}{1}}
\section{Introduction}
This article continues a series of works dealing with the Hartree–Fock approximation of quantum electrodynamics (QED) \cite{GraLewSer-2009-CMP, GraLewSer-2011-CMP, HaiLewSer-2005-CMP, HaiLewSer-2005-JPA, HaiLewSerSol-2007-PRA, HaiLewSol-2007-CPAM} originated from a seminal paper by Chaix and Iracane \cite{ChaIra-1989-JPB}. In these works, only the purely electrostatic case at zero temperature has been considered. Later, the Dirac vacuum interacting also with a magnetic field has been rigorously studied \cite{GraHaiLewSer-2013-ARMA, GraLewSer-2018-JMPA}. As described below, in various situations it is interesting to consider the same kind of problem in a thermal background. A first step in this direction is the work \cite{Mor-2026-AHP} on the purely magnetic case. In the present paper, we extend this analysis to include an electrostatic potential by establishing, under suitable assumptions, the well-posedness and properties of a free energy functional for the Dirac field at thermal equilibrium. This framework is intended as a first step towards a more comprehensive study of screening effects in the polarised equilibrium state, similarly to the analysis carried out in \cite{HaiLewSei-2008-RMP} for the purely electrostatic case within the \emph{reduced Bogoliubov–Dirac–Fock model}. Such an investigation will be the subject of a future work.
\medskip

The non-linear polarisable nature of the vacuum in quantum field theory is a well-established fact since a long time \cite{Dir-1934-MPCPS, EulKoc-1935-DN, Wei-1936-MFM}. In this context, one replaces the classical action with an effective one accounting for quantum corrections, still guaranteeing the validity of the principle of least action. In quantum electrodynamics, a well-known way to implement this strategy is by integrating out fermions in the path integral formulation (see for instance \cite[Ch.~33]{Sch-2013-book}). As a result, we obtain a functional of a classical electromagnetic field treated as an external one. In \cite{GraHaiLewSer-2013-ARMA}, the authors derive in a rigorous way the effective Lagrangian action for time-independent fields in the Coulomb gauge at zero temperature,
\begin{multline*}
    \mathcal{L}\left(\boldsymbol{A}\right)=-\mathcal{F}_\mathrm{vac}\left(e\boldsymbol{A}\right)+e\int_{\R^3}\left(j_\mathrm{ext}\cdot A\left(x\right)-\rho_\mathrm{ext}\left(x\right)V\left(x\right)\right)\diff x\\
    +\frac{1}{8\pi}\int_{\R^3}\left(|E\left(x\right)|^2-|B\left(x\right)|^2\right)\diff x\,,
\end{multline*}
where, in particular, the first term on the right-hand side corresponds to the energy of the vacuum in external electromagnetic fields.
Here $e$ is the elementary charge of an electron, $\boldsymbol{A}=\left(V,A\right)$ is a classical $\R^4-$valued electromagnetic potential with corresponding field $\boldsymbol{F}=\left(E,B\right)=\left(-\nabla V,\mathrm{curl}A\right)$, and $\rho_\mathrm{ext}$ and $j_\mathrm{ext}$ are given external charge and current densities. The corresponding Euler-Lagrange equations are given by the usual Maxwell equations with non-linear and non-local correction terms. Those quantum corrections are negligible at usual energy scales and start to be significant only for extremely intense electromagnetic fields. Indeed, the threshold above which the electromagnetic field starts behaving in a non-linear way is known as \textit{Schwinger limit} and it is several orders of magnitude above what can be currently produced in a laboratory ($E_c\simeq 10^{18}\,V/m$ and $B_c\simeq 10^9\,T$). To give an idea, such an electric field would accelerate a proton from rest to the maximum energy attained by protons at the Large Hydron Collider in only approximately $5$ micrometers. However, the detection of these non-linear effects is still a very active area of experimental research \cite{BurFieHorSpe-1997-PRL, MouTajBul-2006-RMP}. Moreover, it is known that such a magnetic field strength is attained on the surface of magnetars, that is neutron stars with really large magnetic field, where these terms play an important role \cite{BarHar-2001-TAJ, DenSve-2003-AA, MarBroSte-2003-PRL}. In particular, the observation of one of these effects has been recently announced \cite{FanKamInaYam-2017-EPJD, MigTestGonTav-2016-MNRAS}.\par
The study of positive-temperature effects in quantum field theory attracted a considerable interest in the Physics literature in the past decades.
Indeed, in the non-perturbative physical examples above, the fields exist in a thermal bath or in some non-equilibrium background which is very different from the vacuum. Several authors \cite{DolJac-1974-PRD, KirLin-1972-PLB, Wei-1974-PRD} considered systems of elementary particles described by a quantum field theory in a positive-temperature background. They found that while specific symmetries are spontaneously broken at zero temperature (such as those pertaining weak interactions), they may be restored at sufficiently high temperatures. Moreover, the critical temperature characterising such phase transition has been calculated. In order to do this, one needs to know the Feynman rules for a field theory at finite temperature. For a non-gauge theory, they can be derived using well-known methods \cite{FetWal-1971-book}. However, for a gauge theory, various problems arise and a more powerful technique is needed  \cite{Ber-1974-PRD}. We mention that, in any case, the Feynman rules at finite temperature $T$ are strictly related to the Kubo-Martin-Schwinger (KMS) condition (see also Remark~\ref{stat:KMS-condition}) and concretely consist in the zero-temperature rules with the following formal replacements:
\begin{align}\notag
    \int\frac{\diff^4 p}{\left(2\pi\right)^4}&=\frac{i}{\beta}\sum_n\int\frac{\diff^3 p}{\left(2\pi\right)^3}\,,\\\label{eq:T-Feynman-rules}
    p_0&=i\omega_n\,,\\\notag
    \left(2\pi\right)^4\delta^4\left(p_1+p_2+\ldots\right)&=\frac{\beta}{i}\left(2\pi\right)^3\delta\left(\omega_{n_1}+\omega_{n_2}+\ldots\right)\delta^3\left(\Vec{p_1}+\Vec{p_2}+\ldots\right)\,,
\end{align}
where $\beta=1/T$ and
\[
\omega_n=
\begin{cases}
    \frac{2n\pi i}{\beta}\quad\text{for bosons}\,,\\
    \frac{\left(2n+1\right)\pi i}{\beta}\quad\text{for fermions}\,,
\end{cases}
\]
depending on the statistics of the particles involved.\par

We mention that thermal effects in quantum electrodynamics have been first studied in \cite{Dit-1979-PRD}, where the author calculates the one-loop effective potential at finite temperature for scalar and spinor QED in the presence of a constant magnetic field, namely the Euler-Heisenberg Lagrangian at positive temperature. In the language of thermodynamics, such an effective potential represents the contribution of the vacuum energy to the total free energy in presence of an external constant field. Later, the thermal effective action has been analysed in order to deal with finite temperature and chemical potential, under more general assumptions \cite{Dit-1979-PRD,ElmPerSka-1993-PRL, ElmSka-1995-PLB, Sch-1951-PR}. 
We refer the reader to \cite[Chapter~3.5]{DitGie-2000-book} for a complete review about QED at finite temperature.\par
The aim of this paper is to properly define a free energy functional for the Dirac field at thermal equilibrium in a (classical) electromagnetic field, and to study some of its properties. As usual in this context, this quantity suffers from divergences in the high-energy regime which need to be regularised. For this purpose, we employ the Pauli-Villars technique, as done in \cite{GraHaiLewSer-2013-ARMA, Mor-2026-AHP}. In particular, as already mentioned, we extend the analysis of the purely magnetic case \cite{Mor-2026-AHP} to the presence of an additional electrostatic field.\par

To conclude, we observe that QED models of Hartree-Fock type have been applied to the study of relativistic molecular dynamics 
\cite{Mor-2025-JDE}, while various properties of graphene and related systems have been described in \cite{BorMor-2025-JPA, HaiLewSpa-2012-JMP}, formally treated as 2D QED systems. Very recently, a one-dimensional version has been numerically studied, as a first step towards the concrete implementation of an effective QED theory for atoms and molecules \cite{AudMorLevTou-2025-JPA}.
\addtocontents{toc}{\protect\setcounter{tocdepth}{1}}
\subsection*{Organisation of the paper}
We start by introducing the general framework and deriving the model under study in \cref{sec:der}. Then, in \cref{sec:main_thm} we state our main result on the rigorous definition of a Pauli-Villars-regularised free energy functional for the Dirac field at thermal equilibrium, under suitable assumptions on the electromagnetic field. This is the content of \cref{thm:main}, whose proof is given in \cref{sec:free_energy}.

\addtocontents{toc}{\protect\setcounter{tocdepth}{2}}
\subsection{Derivation of the model}\label{sec:der}
Our aim is to rigorously define a free energy functional for the Dirac field in an electromagnetic field at thermal equilibrium. For this purpose, we first introduce the model in the zero-temperature case and recall some basic properties of Dirac operators which will be needed later.
\medskip

We consider a fermionic second-quantised field in a given \emph{classical} electromagnetic potential $\bm A=(V,A)$, where $V$ and $A$ are respectively the electrostatic and the magnetic potential. The Hamiltonian of the field at zero temperature reads

\begin{equation}\label{eq:field-hamiltonian}
    \mathbb{H}^{e\boldsymbol{A}}=\frac{1}{2}\int_{\R^3}\left(\Psi^*\left(x\right)D^{e\boldsymbol{A}}_m\Psi\left(x\right)-\Psi\left(x\right)D^{e\boldsymbol{A}}_m\Psi^*\left(x\right)\right)\diff x\,,
\end{equation}
where $\Psi\left(x\right)$ is the second-quantised field operator, physically interpreted as the annihilation of an electron at $x$ and which satisties the following anticommutation relation:
\[
\Psi^*\left(x\right)_\sigma \Psi\left(y\right)_\nu +\Psi\left(y\right)_\nu \Psi^*\left(x\right)_\sigma =2\delta_{\sigma,\nu}\delta\left(x-y\right)\,.
\]
Here $1\leq\sigma,\nu\leq 4$ are the spin variables and $\Psi\left(x\right)_\sigma$ is an operator-valued distribution. Formally, the Hamiltonian $\mathbb{H}^{e\boldsymbol{A}}$ acts on the fermionic Fock space,
\[
\mathcal{F}=\C\oplus\bigoplus_{N\geq 1}\bigwedge_1^N L^2\left(\R^3,\C^4\right)\,,
\]
and is charge conjugation invariant, meaning that
\[
\operatorname{C}\mathbb{H}^{e\boldsymbol{A}}\operatorname{C}^{-1}=\mathbb{H}^{-e\boldsymbol{A}}\,,
\]
where $\operatorname{C}$ is the charge-conjugation operator in Fock space as defined in \cite[Formula~(1.81)]{Tha-1992-book}. Notice that the integrand in \eqref{eq:field-hamiltonian} is interpreted as follows
\begin{multline*}
    \Psi^*\left(x\right)D^{e\boldsymbol{A}}_m\Psi\left(x\right)-\Psi\left(x\right)D^{e\boldsymbol{A}}_m\Psi^*\left(x\right)\\
    =\sum_{\mu,\nu=1}^4\left(\Psi^*\left(x\right)_\mu \left(D^{e\boldsymbol{A}}_m\right)_{\mu,\nu}\Psi\left(x\right)_\nu-\Psi\left(x\right)_\mu\left(D^{e\boldsymbol{A}}_m\right)_{\mu,\nu}\Psi^*\left(x\right)_\nu\right) \,,
\end{multline*}
where
\[
D^{e\boldsymbol{A}}_m=\boldsymbol{\alpha}\cdot\left(-i\nabla-eA\right)+eV+m\beta
\]
is the electromagnetic Dirac operator acting on $L^2\left(\R^3,\C^4\right)$, for a fermion of mass $m$. We work in a system of units such that the speed of light and the reduced Planck constant are both set equal to one, $c=\hbar=1$. The four Dirac matrices $\boldsymbol{\alpha}=\left(\alpha_1,\alpha_2,\alpha_3\right)$ and $\beta$ are given by
\[
\alpha_k=\begin{bmatrix}
0 & \sigma_k\\
\sigma_k & 0
\end{bmatrix}\,,\quad
\beta=\begin{bmatrix}
I_2 & 0\\
0 & -I_2
\end{bmatrix}\,,
\]
with Pauli matrices $\sigma_1$, $\sigma_2$ and $\sigma_3$ defined as
\[
\sigma_1=\begin{bmatrix}
0 & 1\\
1 & 0
\end{bmatrix}\,,\quad
\sigma_2=\begin{bmatrix}
0 & -i\\
i & 0
\end{bmatrix}\,,\quad
\sigma_3=\begin{bmatrix}
1 & 0\\
0 & -1
\end{bmatrix}\,.
\]
The \emph{free} Dirac operator $D^0_m$, \emph{i.e.} with $\bm A\equiv 0$, is self-adjoint on the Sobolev space $H^1(\R^3,\C^4)$, and its spectrum is purely absolutely continuous and given by
\begin{equation}\label{eq:free-Dirac-spectrum}
\sigma(D^0_m)=(-\infty,-m]\cup[m,+\infty)\,,
\end{equation}
see \cite{Tha-1992-book}. This led Dirac to postulate that the vacuum is filled with infinitely many virtual particles occupying the negative energy states (\textit{Dirac sea}) \cite{Dir-1930-PRSL, Dir-1934-MPCPS, Dir-1934-SR}. As a consequence, a real free electron cannot be in a negative state due to the Pauli exclusion principle. In addition to that, Dirac also conjectured the existence of `holes' in the Dirac sea interpreted as positrons, having a positive charge and a positive energy. Another prediction of the theory was the phenomenon of vacuum polarisation: while the free vacuum is electrically neutral, in presence of an external electrostatic field the virtual electrons respond to the applied field and the vacuum acquires a non-constant density of charge. Actually, the polarised vacuum itself modifies the electrostatic field and the virtual electrons react to the corrected field. We refer the reader to \cite{EstLewSer-2008-BAMS} for a detailed review about the issues arising from the negative part of the spectrum of the free Dirac operator.\par

The electromagnetic field is given by the vector
\[
\boldsymbol{F}=(E,B):=\left(-\nabla V,\mathrm{curl}A\right)\,.
\]
Thus, it is natural to introduce the following Coulomb-gauge homogeneous Sobolev space:
\[
\Cgaugesobolev=\setst{\boldsymbol{A}=\left(V,A\right)\in L^6\left(\R^3,\R^4\right)}{\mathrm{div}A=0\text{ and }\boldsymbol{F}=\left(-\nabla V,\mathrm{curl}A\right)\in L^2\left(\R^3,\R^6\right)}\,,
\]
endowed with its norm
\begin{equation}\label{eq:H1div-norm}
\|\boldsymbol{A}\|^2_{\Cgaugesobolev}=\|\nabla V\|^2_{L^2\left(\R^3\right)}+\|\mathrm{curl}A\|^2_{L^2\left(\R^3\right)}=\|\boldsymbol{F}\|^2_{L^2\left(\R^3\right)}\,.
\end{equation}
The equation $\mathrm{div}A=0$ should be interpreted in a distributional sense and fixes the Coulomb gauge for the potential, while the assumption $\boldsymbol{F}\in L^2\left(\R^3\right)$ means that the electromagnetic field has a finite energy, namely
\[
\int_{\R^3}|E|^2+|B|^2 \,\diff x<\infty\,.
\]
Concerning the basic spectral properties of the electromagnetic Dirac operator $D^{e\boldsymbol{A}}_m$, we state the following lemma, whose proof can be found in \cite[Lemma~2.1]{GraHaiLewSer-2013-ARMA}:
\begin{lem}\label{stat:Dirac-spectral-properties}
    Let $m>0$ and $\boldsymbol{A}\in\Cgaugesobolev$.
    \begin{enumerate}[i)]
        \item The operator $D^{e\boldsymbol{A}}_m$ is self-adjoint on $H^1\left(\R^3,\C^4\right)$ and its essential spectrum is
        \[
        \sigma_\mathrm{ess}\left(D^{e\boldsymbol{A}}_m\right)=\intoc{-\infty}{-m}\cup\intco{m}{\infty}\,.
        \]
        \item There exists a universal constant $C$ such that, if $\|\boldsymbol{A}\|_{\Cgaugesobolev}\leq\eta\sqrt{m}\,$, for some number $\eta<1/C\,$, then
        \[
        \sigma\left(D^{e\boldsymbol{A}}_m\right)\cap\intoo{-m\left(1-C\eta\right)}{\left(1-C\eta\right)m}=\varnothing\,.
        \]
        \item Finally, if $V\equiv0$, then $\sigma\left(D^{e\boldsymbol{A}}_m\right)\cap\intoo{-m}{m}=\varnothing\,$.
    \end{enumerate}
\end{lem}
\begin{rmk}\label{stat:rmk-spectral-gap}
Notice that item $ii)$ ensures the persistence of a spectral gap around $\lambda=0$, under suitable smallness assumptions on the electromagnetic \emph{field} (see \eqref{eq:H1div-norm}). This is consistent with the fact that, for arbitrary field strengths, other physical effects such as electron-positron pair creation can occur. In particular, this is expected to be related to the electrostatic field. Indeed, item $iii)$ shows that the the spectrum \eqref{eq:free-Dirac-spectrum} is stable under purely magnetic perturbations. However, as we shall see later in this work, the discretising effect induced by the temperature, and related to the KMS condition (see also Remark~\ref{stat:KMS-condition}), seems to remove the necessity of a spectral gap at positive temperature, even in the presence of an electrostatic field.
\end{rmk}

The expectation value of the Hamiltonian \eqref{eq:field-hamiltonian} in any state of the Fock space can be expressed as \cite[Equation~(2.8)]{GraHaiLewSer-2013-ARMA}
\begin{equation}\label{eq:expectation-value-energy}
    \langle\mathbb{H}^{e\boldsymbol{A}}\rangle=\tr_{L^2\left(\R^3,\C^4\right)}\left[D^{e\boldsymbol{A}}_m\left(\gamma-\frac{1}{2}\right)\right]=:E\left(\gamma\right)\,,
\end{equation}
where $\gamma$ is the one-particle density matrix associated to a given state, namely
\[
\gamma\left(x,y\right)_{\sigma,\nu}=\langle\Psi^*\left(x\right)_\sigma \Psi\left(y\right)_\nu\rangle\,.
\]
The subtraction of half the identity in \eqref{eq:expectation-value-energy} comes from charge-conjugation invariance, so that the energy is expressed in terms of the renormalised density matrix $\gamma-1/2$. We refer the reader to \cite{HaiLewSol-2007-CPAM} for more details. Being fermions, electrons must obey the Pauli exclusion principle, encoded in the condition 
\begin{equation}\label{eq:Pauli-principle-dm}
0\leq\gamma\leq 1 \qquad\mbox{on} \quad L^2\left(\R^3,\C^4\right)\,.
\end{equation}
 On the other hand, any operator $\gamma$ verifying $0\leq\gamma\leq 1$ formally arises from a unique quasi-free state in the fermionic Fock space. Since the energy only depends on the density matrix $\gamma$ associated with the state, we can reformulate our problem by focusing on the operator $\gamma$ and on the energy \eqref{eq:expectation-value-energy}.\par
\smallskip
Recall that we are interested in studying the Dirac field at thermal equilibrium. Thus, we first introduce the free energy:
\begin{equation}\label{eq:definition-free-energy}
    F\left(\gamma\right)=E\left(\gamma\right)-T \tr\left[S\left(\gamma\right)\right]\,,
\end{equation}
where
\[
S\left(x\right)=-x\log x-\left(1-x\right)\log\left(1-x\right)
\]
and $\tr\left[S\left(\gamma\right)\right]$ is a well-known entropy term studied in the literature (see e.g. \cite[Formula~(2c.10)]{BacLieSol-1994-JSP}, \cite[Formula~(4)]{HaiLewSei-2008-RMP} and \cite[Formula~(5)]{GonKouSer-2023-AHP}). Secondly, we address the minimisation of the free energy functional \eqref{eq:definition-free-energy} with respect to $\gamma$. 
\par


Arguing as in \cite[Lemma~1.1]{GonKouSer-2023-AHP}, one sees that the only critical point $\gamma^*$ of the functional \eqref{eq:definition-free-energy}, which is also the unique minimiser by convexity, is given by
\[
\gamma^*=\frac{1}{1+e^{D^{e\bm A}_m /T}}=\frac{e^{-D^{e\bm A}_m /2T}}{2\cosh\left(D^{e \bm A}_m /2T\right)}\,,\hence 1-\gamma^*=\frac{1}{1+e^{-D^{e\bm A}_m /T}}=\frac{e^{D^{e\bm A}_m /2T}}{2\cosh\left(D^{e\bm A}_m /2T\right)}\,.
\]
Then, the corresponding value of the free energy is formally given by
\begin{align}\notag
    F\left(\gamma^*\right)&:=\tr\left[D^{e\bm A}_m\left(\gamma^*-\frac{1}{2}\right)\right]-T\tr\left[S\left(\gamma^*\right)\right]\\ \label{eq:infinite-free-energy}
    &=-\frac{1}{\beta}\tr\left[\log\left(2\cosh\frac{\beta D^{e\bm A}_m}{2}\right)\right]=:\mathcal{F}\left(e\bm A,\beta\right)\,.
\end{align}
As usual in this context, this free energy is infinite except if our model is settled in a box with an ultraviolet (UV) cut-off \cite{HaiLewSol-2007-CPAM}. In order to rigorously define \eqref{eq:infinite-free-energy}, we proceed as follows. First, we subtract the (infinite) free energy of the free Dirac sea and define the \emph{relative free energy} as
\begin{equation}\label{eq:relative-free-energy}
    \mathcal{F}_\mathrm{rel}\left(e\bm A,\beta\right)=\frac{1}{\beta}\tr\left[\log\left(2\cosh\frac{\beta D^{0}_m}{2}\right)-\log\left(2\cosh\frac{\beta D^{e\bm A}_m}{2}\right)\right]\,.
\end{equation}
From a variational perspective, one can expect that removing an infinite constant \emph{formally} does not alter our problem. However, the functional \eqref{eq:relative-free-energy} is not well-defined yet, due to the the well-known ultraviolet divergences of QED. Indeed, the operator $\log\left(2\cosh\left(\beta D^{0}_m/2\right)\right)-\log\left(2\cosh\left(\beta D^{e\bm A}_m/2\right)\right)$ is not trace-class as long as $\bm A\neq 0$ and this can be formally seen by expanding the trace in a power series of $e\bm A$. As we will see later, the first-order term in the expansion vanishes, while the second-order term is infinite due to ultraviolet divergences. Since the higher-order terms cannot prevent this issue, an UV cut-off is needed.\par
In this regard, the choice of the regularisation is extremely important. As shown below, thanks to the gauge invariance, several terms in the expansion of $\eqref{eq:relative-free-energy}$ as a power series of $e\bm{A}$ vanish. This is a further reason to preserve the gauge symmetry, in addition to its well-known physical meaning. In \cite{GraLewSer-2009-CMP}, the authors studied two approaches to handle UV divergences in the purely electrostatic case by imposing a sharp cut-off. However, these methods are not applicable in our setting, as they break gauge symmetry.\par

The regularisation technique we are going to use has been introduced by Pauli and Villars \cite{PauVil-1949-RMP} in $1949$. Of course, this is not the only possible choice (see for instance \cite{Lei-1975-RMP} for an alternative approach). We remark that the same approach has been used in the previous works \cite{GraHaiLewSer-2013-ARMA, GraLewSer-2018-JMPA} at zero temperature and in \cite{Mor-2026-AHP} for the purely magnetic case at positive temperature.

The mentioned Pauli-Villars (PV) method consists in introducing $J$ fictitious particles into the model with very high masses $m_1,\ldots,m_J$, essentially playing the role of ultraviolet cut-offs. Indeed, note that $m$ has the dimension of a momentum since in our system of units $\hbar=c=1$. The sole purpose of these additional particles is to regularise the model at high energies, and they have no direct physical interpretation. Moreover, being very massive, they do not affect the low-energy regime. In our framework, the method consists in considering the following energy functional:
\begin{equation}\label{eq:PV-free-energy}
    \widetilde{\mathcal{F}}_\mathrm{PV}\left(e\bm A,\beta\right):=\frac{1}{\beta}\tr\left[\sum_{j=0}^J c_j\left(\log\left(2\cosh\frac{\beta D^{0}_{m_j}}{2}\right)-\log\left(2\cosh\frac{\beta D^{e\bm A}_{m_j}}{2}\right)\right)\right]\,.
\end{equation}
We take $m_0=m$ to be the `real' mass, and $c_0=1$. The coefficients $c_j$ and $m_j$ are chosen so that
\begin{equation}\label{eq:general-PV-conditions}
    \sum_{j=0}^J c_j=\sum_{j=0}^J c_j m_j^2=0.
\end{equation}
It is well known in the Physics literature \cite{BjoDre-1965-book,GreRei-2009-book,PauVil-1949-RMP} that only two auxiliary fields are required to satisfy these conditions. Accordingly, we set $J=2$ throughout the remainder of this work in order to simplify the presentation. Thus, the condition \eqref{eq:general-PV-conditions} is equivalent to
\[
c_1=\frac{m_0^2-m_2^2}{m_2^2-m_1^2}\et c_2=\frac{m_1^2-m_0^2}{m_2^2-m_1^2}\,.
\]
Furthermore, we shall always assume that $m_0<m_1<m_2$, which implies that $c_1<0$ and $c_2>0$. The role of the constraint \eqref{eq:general-PV-conditions} is to remove the worst \emph{linear} ultraviolet divergences. Indeed, the Pauli-Villars regularisation still does not avoid a logarithmic divergence when one takes $m_1,m_2\xrightarrow{}\infty$. This can be understood defining the averaged ultraviolet cut-off $\Lambda$ as
\begin{equation}\label{eq:averaged-UV-cutoff}
    \log\left(\Lambda^2\right)=-\sum_{j=0}^2 c_j\log\left(m_j^2\right)\,.
\end{equation}
Observe that, once the value of $\Lambda$ is fixed, the values of $m_1,m_2$ can not be uniquely determined. Concretely one usually chooses the masses as functions of the cut-off $\Lambda$, so that the coefficients $c_1,c_2$ stay bounded when $\Lambda\xrightarrow{}\infty$. As already mentioned, the persistent logarithmic divergence in the averaged cut-off $\Lambda$ can explicitly be seen in the second-order term in the expansion and will be studied later in this work (see item $c)$ in Remark~\ref{rmk:main}).  \par
In order to simplify the calculations, it is convenient to rewrite the Pauli-Villars-regularised free energy functional \eqref{eq:PV-free-energy} in an integral form. To this aim, we differentiate it with respect to $\beta$
\[
\frac{\partial}{\partial\beta}\left(\beta \widetilde{\mathcal{F}}_\mathrm{PV}\left(e\bm A,\beta\right)\right)=\frac{1}{2}\tr\left[\sum_{j=0}^2 c_j \left(D^{0}_{m_j} \tanh\frac{\beta D^{0}_{m_j}}{2}-D^{e\bm A}_{m_j} \tanh\frac{\beta D^{e\bm A}_{m_j}}{2}\right)\right]\,,
\]
which gives
\begin{equation}\label{eq:integral-free-energy}
    \widetilde{\mathcal{F}}_\mathrm{PV}\left(e\bm A,\beta\right)=\frac{1}{2\beta}\int_0^\beta \tr\left[\sum_{j=0}^2 c_j \left(D^{0}_{m_j} \tanh\frac{b D^{0}_{m_j}}{2}-D^{e\boldsymbol{A}}_{m_j} \tanh\frac{b D^{e\bm A}_{m_j}}{2}\right)\right]\diff b\,.
\end{equation}
Notice that, using conditions \eqref{eq:general-PV-conditions}, it is not hard to see that the value of $\widetilde{\mathcal{F}}_\mathrm{PV}$ goes to zero as $\beta$ goes to zero. However, we still need to modify the free energy functional. Indeed, following the study of the electromagnetic case at zero temperature \cite{GraHaiLewSer-2013-ARMA}, we will rather work with
\begin{equation}\label{eq:boxed-PV-free-energy}
\boxed{
 \mathcal{F}_\mathrm{PV}\left(e\bm A,\beta\right)=\frac{1}{2\beta}\int_0^\beta \tr\left(\tr_{\C^4}\left[\sum_{j=0}^2 c_j \left(D^{0}_{m_j} \tanh\frac{b D^{0}_{m_j}}{2}-D^{e\boldsymbol{A}}_{m_j} \tanh\frac{b D^{e\bm A}_{m_j}}{2}\right)\right]\right)\diff b\,.}
\end{equation}
Notice that, compared to the previous expression, we have inserted a $\C^4-$trace in the integrand. As remarked in \cite{GraHaiLewSer-2013-ARMA}, this ensures that the resulting operator is indeed trace-class in presence of the electrostatic potential. Alternatively, increasing the number of fictitious particles in the Pauli-Villars scheme would allow to remove problematic terms, imposing additional constraints on the masses. Nevertheless, we prefer to modify the functional and to work with two auxiliary masses, as typically done in the Physics literature.

\begin{rmk}
    We emphasise that the equilibrium state considered in this paper corresponds to a system at positive temperature $T > 0$. While some mathematical physics literature occasionally refers to such a generic equilibrium state as `vacuum state' or `Dirac vacuum', such as in \cite{HaiLewSei-2008-RMP, Mor-2026-AHP}, in the Physics literature the `vacuum' strictly denotes the zero-temperature ground state. To avoid ambiguity, throughout this paper we consistently refer to this state as the `thermal equilibrium state' (or the `Dirac field at thermal equilibrium').
\end{rmk}

\medskip
Let us stress that in the rest of the paper the elementary charge $e$ will be set to be equal to $1$, for the sake of simplicity.

\subsection{Main theorem}\label{sec:main_thm}
Our main result shows that the free energy functional $\bm A\mapsto \FF(\bm A,\beta)$ in \eqref{eq:boxed-PV-free-energy} is well-defined for a general $4$--potential $\bm A=(V,A)$ in the energy space $\hd$ (where we fix the Coulomb gauge condition $\dd A=0$). To further study its properties, we also require the electric field to be integrable, $E=-\nabla V\in L^1(\R^3)$, which is clearly a gauge-invariant assumption.

\begin{thm}[Rigorous definition of $\FF$]\label{thm:main}
Assume that the constants $c_j$ and $m_j$ verify
\[
c_0=1\,,\quad m_2>m_1>m_0>0
\]and
\begin{equation}\label{eq:PV-conditions}
\quad \sum^2_{j=0}c_j=\sum^2_{j=0}c_jm^2_j=0\,.
\end{equation}
Let
\begin{equation}\label{eq:TA}
T_{\bm A}(\beta)=\frac{1}{2}\sum^2_{j=0}c_j\left( D^0_{m_j}\tanh{\left(\frac{\beta D^0_{m_j}}{2} \right)}- D^{ \bm A}_{m_j}\tanh{\left(\frac{\beta D^{\bm A}_{m_j}}{2}\right)} \right)\,,
\end{equation}
with $\bm A=(V,A)$. Then,
\begin{enumerate}[i)]
\item For every $\bm A\in L^1(\R^3,\R^4)\cap H^1(\R^3,\R^4)$, the operator $\tr_{\C^4}T_{\bm A}(\beta)$ is trace-class on $L^2(\R^3,\C)$. In particular, the functional $\FF(\bm A,\beta)$ is well-defined as
\begin{equation}\label{eq:FPV}
\boxed{
\FF(\bm A,\beta)=\frac{1}{\beta}\int^\beta_0 \tr\left(\tr_{\C^4}T_{\bm A}(b) \right)\diff b\,.}
\end{equation}
\item Let $\bm{A}\in L^1(\R^3,\R^4)\cap\hd$ and $E=-\nabla V\in L^1(\R^3,\R^3)$. Then, we have
\begin{equation}\label{eq:FPV-decomposition}
\boxed{
\FF(\bm A,\beta)=\mathcal F_2(\bm F,\beta)+\mathcal R(\bm A,\beta)\,,
}
\end{equation}
where $\bm F=(E,B)$, with $B=\curl A$. The functional $\mathcal R$ is continuous on $\hd$ and there exists a constant $\kappa=\kappa(\beta)$ such that 
\begin{equation}\label{eq:R-bound}
\vert \mathcal R(\bm A,\beta)\vert \leq \kappa\left(\left(\sum_{j=0}^2\frac{\abs{c_j}}{m_j}\right)\norm{\bm{F}}^4_{L^2\left(\R^3\right)}+\left(\sum_{j=0}^2\frac{\abs{c_j}}{m_j^2}\right)\norm{\bm{F}}^6_{L^2\left(\R^3\right)}\right)\,.
\end{equation}
\item The functional $\mathcal F_2$ is the bounded quadratic form on $L^1(\R^3,\R^4) \cap  L^2(\R^3,\R^4)$ given by
\begin{empheq}[box=\fbox]{align}   \label{eq:F2}
\begin{split}
    \mathcal F_2(\bm F,\beta)=\frac{1}{8\pi}\int_{\R^3}\left(M^0(k)+M^T(k,\beta)\right)&\left(\vert \widehat B(k)\vert^2-\vert \widehat E(k)\vert^2\right)\, \diff k\\
    &\quad\quad+ \frac{1}{\beta}\int^\beta_0\int_{\R^3}\frac{\Gamma(k,b)}{\vert k\vert^2}\vert \widehat E(k)\vert^2\,\diff k\,\diff b\,,
    \end{split}
\end{empheq}
with
\begin{equation}\label{eq:Mzero}
M^0(k)=-\frac{2}{\pi}\int^1_0\sum^2_{j=0}c_j\log\left(m^2_j+u(1-u)k^2\right)\, u(1-u)\,du\,,
\end{equation}
\begin{equation}\label{eq:MT}
\begin{split}
M^T(k,\beta)&=-\frac{8}{\pi}\int^1_0 du\, u(1-u)\int^\infty_0\,dt\,\sum^2_{j=0}c_j \Bigg(\frac{1}{1+e^{-X_j(\beta,u,k)\cosh t} } \\
&+ \frac{1}{X_j(\beta,u,k)\cosh t}\left(2\log2-\log\left(\left(1+e^{X_j(\beta,u,k)\cosh t}\right) \left(1+e^{-X_j(\beta,u,k)\cosh t}\right) \right) \right)\Bigg)\,,
\end{split}
\end{equation}
where $X_j\left(\beta,u,k\right)=\beta\sqrt{m_j^2+u\left(1-u\right)k^2}$ for $j=0,1,2$, and $\Gamma \in C^0(\R^3\times(0,+\infty))$ is an explicit function of $(k,\beta)$ (see \eqref{eq:Gamma_sum} and the discussion before Lemma~\ref{lem:F2_proof}), such that, for any $\beta\in\intoo{0}{+\infty}$,
\begin{equation}\label{eq:Gamma-properties-thm}
\frac{\Gamma(\cdot,\beta)}{\vert \cdot\vert^2}\in L^1_{\mathrm{loc}}(\R^3)\,,\quad \mbox{and}\quad \frac{\Gamma(\cdot,\beta)}{\vert \cdot\vert^2}\in L^\infty(\{\vert k\vert\geq \delta\})\,,\quad\forall \delta>0\,.
\end{equation}
\item The functional $\FF$ can be uniquely extended to a continuous mapping on
\begin{equation}\label{eq:X}
X:=\{\bm A=(V,A)\in \hd\,:\, E=-\nabla V\in L^1(\R^3,\R^3)\}\,.
\end{equation}
\end{enumerate}
\end{thm}
Some comments on the previous result are listed below.
\begin{rmk}\label{rmk:main}\textcolor{white}{...}
\begin{enumerate}[a)]
\item We first establish in item $i)$ the well-posedness of the PV-regularised free energy functional under non gauge-invariant assumptions. Then, in \eqref{eq:FPV-decomposition} we show that \eqref{eq:FPV} can be rewritten as the sum of a principal term involving the \emph{fields}, with a remainder depending on the potential but which can be suitably estimated in terms of the electromagnetic field (see \eqref{eq:R-bound}).
\item The quadratic term $\mathcal F_2$, rewritten as in \eqref{eq:F2}, describes the linear response of the Dirac field to the external electromagnetic field at thermal equilibrium.  
\item The multiplier $M^0(k)$ represents the zero-temperature contribution to this linear response, calculated in \cite{GraHaiLewSer-2013-ARMA}. It satisfies
\[
\lim_{\Lambda\to\infty}\left(\frac{2 \log\Lambda}{3\pi}- M^0(k) \right)=U(k):=\frac{\vert k\vert^2}{4\pi}\int^1_0\frac{z^2-z^4/3}{1+\vert k\vert^2(1-z^2)/4}\,\diff z\,,
\]
where the function on the right-hand side was first computed in the Physics literature by Serber \cite{Ser-1935-PR} and Uehling \cite{Ueh-1935-PR}.
\item The thermal contributions to the linear response are given by the multipliers $M^T(k,\beta)$ and $\Gamma(k,\beta)$. The former has been computed in \cite{Mor-2026-AHP} for the purely magnetic case, while the latter is calculated in the present paper. Notice that this last term, contrary to the other two, has a potential singularity at zero momentum (see item $iii)$ in \cref{thm:main}). This fact suggests that the presence of the electrostatic field might induce polarisation in the equilibrium state, leading to a screening effect (\emph{Debye screening}) for the charge generating the external field. Indeed, this fact has been proved in the purely electrostatic case in \cite{HaiLewSei-2008-RMP}, and it is expected to occur also in the presence of an additional magnetic field. This problem will be investigated in a future work.
\item Observe that the finiteness of the last integral in \eqref{eq:F2} follows combining \eqref{eq:Gamma-properties-thm} and the observation that $E\in L^1(\R^3)\cap L^2(\R^3)$ gives $\widehat E\in C^0(\R^3)\cap L^2(\R^3)$.
\end{enumerate}
\end{rmk}

\addtocontents{toc}{\protect\setcounter{tocdepth}{2}}
\section{Rigorous definition of the free energy}\label{sec:free_energy}
The proof of \cref{thm:main} follows a strategy similar to that employed in the analogous results in \cite{GraHaiLewSer-2013-ARMA,Mor-2026-AHP}, and is organised into several intermediate results presented in the following subsections. We refer to those works for the parts of the argument that are very similar, repeating or only sketching them, when necessary, for the reader convenience. Our main focus will be on the analysis of the electrostatic terms.

\subsection{Integrable four-potentials: proof of item i) in \cref{thm:main}}\label{sec:proof_item_i)}
\begin{proof}[Proof of item i) in \cref{thm:main}]
Assume that
\begin{equation}\label{eq:integrable-potential}
\bm A\in L^1(\R^3,\R^4)\cap H^1(\R^3,\R^4)\,.
\end{equation}
We start from the definition of $T_{\bm A}$ in \eqref{eq:TA}, in order to prove that $\FF$ \eqref{eq:FPV} is well-defined. For this purpose, recall the following formula
\begin{align}
    \notag x\tanh{x}&=\sum_{l\in\Z}\frac{4x^2}{\left(2l-1\right)^2\pi^2+4x^2}\\\label{eq:xtanhx}
    &=\frac{1}{2}\sum_{l\in\Z}\left(2-\frac{i\left(2l-1\right)\pi}{2x+i\left(2l-1\right)\pi}+\frac{i\left(2l-1\right)\pi}{2x-i\left(2l-1\right)\pi}\right)\,.
\end{align}
Using the functional calculus for a self-adjoint operator $S$ on $L^2\left(\R^3,\R^4\right)$ with domain $\operatorname{dom}\left(S\right)$, one gets
\[
S\tanh{S}=\frac{1}{2}\sum_{l\in\Z}\left(2-\frac{i\left(2l-1\right)\pi}{2S+i\left(2l-1\right)\pi}+\frac{i\left(2l-1\right)\pi}{2S-i\left(2l-1\right)\pi}\right)\,.
\]
Observe that this series is convergent, when considered as an operator from $\operatorname{dom}\left(S^2\right)$ to the ambient Hilbert space, as
\[
\norm{\frac{4S^2}{4S^2+\left(2l-1\right)^2\pi^2}}_{D\left(S^2\right)\xrightarrow{}L^2\left(\R^3,\R^4\right)}\leq\min\left\{1,\left(\left(2l-1\right)\pi\right)^{-2}\norm{4S^2}_{D\left(S^2\right)\xrightarrow{}L^2\left(\R^3,\R^4\right)}\right\}\,.
\]
Since the domains of $\left(D^0_{m_j}\right)^2$ and $\left(D^{\bm A}_{m_j}\right)^2$ are both equal to $H^2\left(\R^3,\C^4\right)$, we can write
\begin{multline}\label{eq:TA-resolvent-difference}
    \tr_{\C^4}T_{\bm A}\left(\beta\right)=\frac{1}{2\beta}\sum_{l\in\Z}\sum_{j=0}^2 c_j\tr_{\C^4}\left(\frac{i\omega\left(l,\beta\right)}{D^{\bm A}_{m_j}+i\omega\left(l,\beta\right)}-\frac{i\omega\left(l,\beta\right)}{D^{\bm A}_{m_j}-i\omega\left(l,\beta\right)}\right.\\
    \left.-\frac{i\omega\left(l,\beta\right)}{D^0_{m_j}+i\omega\left(l,\beta\right)}+\frac{i\omega\left(l,\beta\right)}{D^0_{m_j}-i\omega\left(l,\beta\right)}\right)\,,
\end{multline}
on $H^2\left(\R^3,\C^4\right)$, setting
\[
\omega\left(l,\beta\right)=\frac{\left(2l-1\right)\pi}{\beta}\,,\qquad l\in\mathbb Z\,.
\]
Notice that $\omega\left(l,\beta\right)\neq 0$, $\forall l\in\Z$, and therefore we do not need any assumption on the gap of the electromagnetic Dirac operator $D^{\bm A}_m$ at this stage. This seems to be a regularising effect due to the positive temperature, which makes the $\omega-$parameter discrete and always non-zero (see Remark~\ref{stat:rmk-spectral-gap}).\\
In order to prove the theorem, we need to show that the right-hand side of \eqref{eq:TA-resolvent-difference} defines a trace-class operator, whose Schatten norm can be estimated as follows:
\begin{multline*}
    \sum_{l\in\Z}\left\lVert\sum_{j=0}^2 c_j \tr_{\C^4}\left(\frac{i\omega\left(l,\beta\right)}{D^{\bm A}_{m_j}+i\omega\left(l,\beta\right)}-\frac{i\omega\left(l,\beta\right)}{D^{\bm A}_{m_j}-i\omega\left(l,\beta\right)}\right.\right.\\
    \left.\left.-\frac{i\omega\left(l,\beta\right)}{D^0_{m_j}+i\omega\left(l,\beta\right)}+\frac{i\omega\left(l,\beta\right)}{D^0_{m_j}-i\omega\left(l,\beta\right)}\right)\right\rVert_{\schatten{1}}=\bigO{\beta}\,,
\end{multline*}
as $\beta\xrightarrow{}0$. This can be proved under the assumption \eqref{eq:integrable-potential}, as at this stage the Coulomb gauge condition $\mathrm{div}\,A=0$ is not needed.

In order to do so, it suffices to estimate the operator, at fixed $l\in\mathbb Z$,
\begin{multline}\label{eq:cal-R-operator}
    R\left(\omega\left(l,\beta\right),\bm A\right):=\sum_{j=0}^2 c_j \tr_{\C^4}\left(\frac{i\omega\left(l,\beta\right)}{D^{\bm A}_{m_j}+i\omega\left(l,\beta\right)}-\frac{i\omega\left(l,\beta\right)}{D^{ \bm A}_{m_j}-i\omega\left(l,\beta\right)}\right.\\
    \left.-\frac{i\omega\left(l,\beta\right)}{D^0_{m_j}+i\omega\left(l,\beta\right)}+\frac{i\omega\left(l,\beta\right)}{D^0_{m_j}-i\omega\left(l,\beta\right)}\right)\,,
\end{multline}
and then sum up the results. Therefore, by using the second resolvent identity
\[
\frac{i\omega\left(l,\beta\right)}{D^{ \bm A}_{m_j}+i\omega\left(l,\beta\right)}-\frac{i\omega\left(l,\beta\right)}{D^0_{m_j}+i\omega\left(l,\beta\right)}=\frac{i\omega\left(l,\beta\right)}{D^{\bm A}_{m_j}+i\omega\left(l,\beta\right)}\left(\boldsymbol{\alpha}\cdot A-V\right)\frac{1}{D^0_{m_j}+i\omega\left(l,\beta\right)}\,,
\]
and iterating it six times, we get
\begin{multline}\label{eq:R-6-terms}
    R\left(\omega\left(l,\beta\right),\bm{A}\right)=\sum_{n=1}^5 \tr_{\C^4}\left(R_n\left(\omega\left(l,\beta\right),\bm{A}\right)+R_n\left(-\omega\left(l,\beta\right),\bm{A}\right)\right)\\
    +\tr_{\C^4}\left(R'_6\left(\omega\left(l,\beta\right),\bm{A}\right)+R'_6\left(-\omega\left(l,\beta\right),\bm{A}\right)\right)\,,
\end{multline}
with
\begin{equation}\label{eq:R-n}
    R_n\left(\omega\left(l,\beta\right),\bm A\right)=\sum_{j=0}^2 c_j\frac{i\omega\left(l,\beta\right)}{D^0_{m_j}+i\omega\left(l,\beta\right)}\left(\left(\boldsymbol{\alpha}\cdot A-V\right)\frac{1}{D^0_{m_j}+i\omega\left(l,\beta\right)}\right)^n
\end{equation}
and
\begin{equation}\label{eq:R'-6}
    R'_6\left(\omega\left(l,\beta\right),\bm A\right)=\sum_{j=0}^2 c_j\frac{i\omega\left(l,\beta\right)}{D^{\bm A}_{m_j}+i\omega\left(l,\beta\right)}\left(\left(\boldsymbol{\alpha}\cdot A-V\right)\frac{1}{D^0_{m_j}+i\omega\left(l,\beta\right)}\right)^6\,.
\end{equation}
As already remarked in \cite{Mor-2026-AHP}, by replacing the integral with respect to $\omega$ with the series with respect to $l\in\Z$, the same proof as in \cite[Proposition~3.1]{GraHaiLewSer-2013-ARMA} holds yielding
\begin{multline}\label{eq:TA-trace-estimate}
    \sum_{l\in\Z}\left(\sum_{n=1}^5\norm{\tr_{\C^4}R_n\left(\omega\left(l,\beta\right),\bm A\right)+\tr_{\C^4}R_n\left(-\omega\left(l,\beta\right),\bm A\right)}_{\schatten{1}}\right.\\
    \left.+\norm{\tr_{\C^4}R'_6\left(\omega\left(l,\beta\right),\bm A\right)+\tr_{\C^4}R'_6\left(-\omega\left(l,\beta\right),\bm A\right)}_{\schatten{1}}\right)=\bigO{\beta},
\end{multline}
as $\beta\xrightarrow{}0$, where the $\bigO{\beta}$ term above depends on the four-potential $\bm A$ through its $L^p-$norms, for $p\in\intinteger{1}{5}$, as well as the $L^2-$norm of the associated electromagnetic field. We remark that the arguments in \cite[Sec. 3]{GraHaiLewSer-2013-ARMA} show that the $\schatten{1}$--norms in \eqref{eq:TA-sum} are actually finite. In particular, the presence of the $\C^4$--trace ensures the finiteness of the $n=1$ and $n=2$ terms (see \cite[Equations~(3.28)~and~(3.32)]{GraHaiLewSer-2013-ARMA}).
\end{proof}
\subsection{Gauge-invariant estimates}
The proof in the previous subsection relies on the non-gauge-invariant assumption \eqref{eq:integrable-potential}, as does the final estimate in \eqref{eq:TA-trace-estimate}, which involves $L^p-$norms of the $4$--potential. For this reason, we now aim to improve the result just proved by recovering gauge invariance. Our starting point is the following formula, obtained combining \eqref{eq:TA-resolvent-difference}, \eqref{eq:cal-R-operator} and \eqref{eq:R-6-terms}: 
\begin{equation}\label{eq:TA-sum}
    \begin{split}
    T_{\bm A}\left(\beta\right)=\sum_{n=1}^5 T_n\left(\bm A,\beta\right)+T'_6\left(\bm A,\beta\right)\coloneqq &\frac{1}{2\beta}\sum_{n=1}^5\sum_{l\in\Z}\left(R_n\left(\omega\left(l,\beta\right),\bm A\right)+R_n\left(-\omega\left(l,\beta\right),\bm A\right)\right)\\
    &+\frac{1}{2\beta}\sum_{l\in\Z}\left(R'_6\left(\omega\left(l,\beta\right),\bm A\right)+R'_6\left(-\omega\left(l,\beta\right),\bm A\right)\right)\,,
\end{split}
\end{equation}
with $R_n$, $R'_6$ given by \eqref{eq:R-n} and \eqref{eq:R'-6}.
Correspondingly, we first take the $\C^4-$trace, to get
\begin{equation}\label{eq:4trace}
    \tr_{\C^4}T_{\bm A}\left(\beta\right)=\sum_{n=1}^5 \tr_{\C^4}T_n\left(\bm A,\beta\right)+\tr_{\C^4}T'_6\left(\bm A,\beta\right)\,.
\end{equation}\par

We shall prove that the operators $\tr_{\C^4}T_n\left(\bm A,\beta\right)$ and $\tr_{\C^4}T'_6\left(\bm A,\beta\right)$ are trace-class, with estimates involving only suitable norms of the electromagnetic field $\bm F=(-\nabla V, \curl A)$. More precisely, we will first work with assumptions which do not preserve the gauge invariance, showing how to remove them by a density argument in the last part of the proof. We start observing that the integrals corresponding to odd-order terms in \eqref{eq:4trace} vanish, due to the charge-conjugation invariance. This fact is usually referred to as Furry's theorem in the Physics literature (see \cite{Fur-1937-PR} and \cite[Sec.~4.1]{GreRei-2009-book}). Then, we will focus on terms of order $4$ and $6$, while a detailed analysis of the quadratic term is provided in \cref{sec:the-quadratic-term}.

\begin{lem}\label{stat:furry-odd-order}
    For $\bm A\in L^1\left(\R^3,\R^4\right)\cap H^1\left(\R^3,\R^4\right)$ and $n=1,3,5$, we have
    \[\begin{split}
        \frac{1}{\beta}&\int_0^\beta\tr\left(\tr_{\C^4}\left[T_n\left(\bm A,b\right)\right]\right)\diff b \\
        &=\frac{1}{2\beta}\int_0^\beta\sum_{l\in\Z}\tr\left(\tr_{\C^4}\left[R_n\left(\omega\left(l,b\right),\bm A\right)+R_n\left(-\omega\left(l,b\right),\bm A\right)\right]\right)\frac{\diff b}{b}=0\,.
        \end{split}
    \]
\end{lem}
\proof
The proof is the same as in \cite[Lemma~4.1]{GraHaiLewSer-2013-ARMA}.
\endproof
Thus, thanks to Lemma~\ref{stat:furry-odd-order}, the Pauli-Villars-regularised free energy functional defined in \eqref{eq:FPV} can be simplified as
\begin{align*}
    \FF(\bm A,\beta)&=\frac{1}{\beta}\int^\beta_0 \tr\left(\tr_{\C^4}T_{\bm A}(b) \right)\diff b\\
    &=\frac{1}{\beta}\int_0^\beta\tr\left(\sum_{n=1}^5 \tr_{\C^4}T_n\left(\bm A,b\right)+\tr_{\C^4}T'_6\left(\bm A,b\right)\right)\diff b\\
    &=\frac{1}{\beta}\int_0^\beta\tr\left(\tr_{\C^4}T_2\left(\bm A,b\right)+\tr_{\C^4}T_4\left(\bm A,b\right)+\tr_{\C^4}T'_6\left(\bm A,b\right)\right)\diff b,,
\end{align*}
which can be rewritten as
\begin{equation}\label{eq:PVs}
    \FF\left(\bm{A},\beta\right)=\mathcal{F}_2\left(\bm{A},\beta\right)+\mathcal{R}\left(\bm{A},\beta\right)\,,
\end{equation}
where 
\begin{equation}\label{eq:F2s}
    \mathcal{F}_2\left(\bm{A},\beta\right)=\frac{1}{\beta}\int_0^\beta\tr\left(\tr_{\C^4}T_2\left(\bm A,b\right)\right)\diff b
\end{equation}
and
\begin{equation}\label{eq:Rs}
    \mathcal{R}\left(\bm{A},\beta\right)=\frac{1}{\beta}\int_0^\beta\tr\left(\tr_{\C^4}T_4\left(\bm A,b\right)+\tr_{\C^4}T'_6\left(\bm A,b\right)\right)\diff b\,.
\end{equation}
\cref{sec:the-quadratic-term} will be devoted to the study of the linear response term $\mathcal{F}_2$ in \eqref{eq:PVs}, which will turn out to depend only on the electromagnetic \emph{field} $\bm{F}$. Instead, the rest of this subsection will focus on finding gauge invariant estimates for the remainder term $\mathcal{R}$. This is the content of the following lemma.
\begin{lem}\label{lem:R_est}
    The functional $\mathcal{R}$, defined in \eqref{eq:Rs}, is continuous on $\Cgaugesobolev$ and there exists a universal constant $\kappa=\kappa\left(\beta\right)$ such that
    \begin{equation}\label{eq:remainder-estimate}
        \abs{\mathcal{R}\left(\bm{A},\beta\right)}\leq \kappa\left(\left(\sum_{j=0}^2\frac{\abs{c_j}}{m_j}\right)\norm{\bm{F}}^4_{L^2\left(\R^3\right)}+\left(\sum_{j=0}^2\frac{\abs{c_j}}{m_j^2}\right)\norm{\bm{F}}^6_{L^2\left(\R^3\right)}\right)\,.
    \end{equation}
\end{lem}
\begin{proof}
     Almost the same proofs as in \cite[Lemma~4.4]{GraHaiLewSer-2013-ARMA} and in \cite[Lemma~4.5]{GraHaiLewSer-2013-ARMA} hold. In order to include a finite temperature, as previously remarked in \cite{Mor-2026-AHP}, recall that we simply have to make the following formal replacements:
    \begin{align*}
        \frac{1}{4\pi}\int_\R\diff\omega&\mapsto\frac{1}{2\beta}\sum_{l\in\Z}\,,\\
        \omega\in\R&\mapsto\omega\left(l,\beta\right)=\frac{\left(2l-1\right)\pi}{\beta},\,l\in\Z\,,
    \end{align*}
    where $\beta=1/T\in\intoo{0}{\infty}$, and then to consider the averaged primitive function with respect to $\beta$ of each term in the expansion. After the substitutions, it is important to notice that the universal constant $K$ in \cite[Lemma~4.4]{GraHaiLewSer-2013-ARMA} and \cite[proof of Lemma~4.5]{GraHaiLewSer-2013-ARMA} translates into a positive-temperature constant which is $\bigO{\beta}$ as $\beta\xrightarrow{}0$ and this is crucial in order to get the finite constant $\kappa$ in \eqref{eq:remainder-estimate} (recall that one has to pay attention to the singularity of the factor $1/\beta$ because of the integration around zero). Indeed, after replacing the integrals over frequencies by series, for instance for the fourth-order term, the positive-temperature constant depends on
    \begin{align*}
        \sum_{j=0}^2 \abs{c_j} \sum_{l\in\Z}\abs{\frac{\omega\left(l,\beta\right)^2}{\left(m_j^2+\omega\left(l,\beta\right)^2\right)^2}}&=\beta^2 \sum_{j=0}^2 \abs{c_j} \sum_{l\in\Z}\abs{\frac{\left(2l-1\right)^2 \pi^2}{\left(\beta^2 m_j^2+\left(2l-1\right)^2\pi^2\right)^2}}\\
        &\leq \frac{\beta^2}{\pi^2} \sum_{j=0}^2 \abs{c_j} \sum_{l\in\Z}\abs{\frac{1}{\left(2l-1\right)^2}}\leq C\beta^2\,,
    \end{align*}
    which clearly shows the constant to be $\bigO{\beta}$ as $\beta\to 0$. A similar argument works for the sixth-order term.
\end{proof}

\subsection{The quadratic term in \eqref{eq:FPV-decomposition}}\label{sec:the-quadratic-term}

We now compute exactly the second-order term $\mathcal{F}_{2}$ associated to the term $T_2\left(\bm A,\beta\right)$ appearing in the decomposition of $T_{\bm A}\left(\beta\right)$ in \eqref{eq:TA-sum}. Our analysis is the positive-temperature counterpart of that in \cite{GraHaiLewSer-2013-ARMA}, and extends the corresponding result in \cite{Mor-2026-AHP} in the presence of an additional electrostatic field. In what follows, we focus only on the new terms with respect to \cite{Mor-2026-AHP}, in order to simplify the presentation. Our ultimate goal in this section is to prove item $iii)$ in \cref{thm:main} and, in particular, to compute the multiplier $\Gamma(k,\beta)$ appearing in \eqref{eq:F2}.\par

Before proceeding with the calculations, we recall some preliminary facts. First of all, as in \cite{Mor-2026-AHP} and as already noticed in the proof of Lemma~\ref{lem:R_est}, compared to the calculations at zero temperature in \cite{GraHaiLewSer-2013-ARMA}, here we make the following formal replacements:
\begin{equation}\label{eq:replacements}
\begin{split}
\frac{1}{4\pi}\int_\R\,\diff\omega&\mapsto\frac{1}{2\beta}\sum_{l\in\mathbb Z}\,,\\
\R\ni\omega&\mapsto \omega(l,\beta)=\frac{(2l-1)\pi}{\beta}\,,\qquad l\in\mathbb Z\,,
\end{split}
\end{equation}
where $\beta=1/T\in(0,\infty)$, and then consider the averaged primitive with respect to $\beta$ of each term appearing in the expansion.
\begin{rmk}\label{stat:KMS-condition}
The above substitutions are related to the well-known KMS condition for quantum statistical systems at thermal equilibrium, prescribing anti-periodicity with respect to an imaginary time for Green's functions. In our case, these replacements are introduced by a different argument, purely mathematical, following from the representation formula \eqref{eq:xtanhx}.
\end{rmk}
Assume that 
\[
\bm A\in L^2(\R^3,\R^4) \cap \hd\,.
\]
Notice that, while the Coulomb gauge condition $\mathrm{div}A=0$ was not essential in the previous subsections , this will be the case in what follows. In view of \eqref{eq:TA-sum}, we have
\[
T_2(\bm A,\beta)=\frac{1}{2\beta}\sum_{l\in\Z}\left(R_2(\omega(l,\beta),\bm A)+R_2(-\omega(l,\beta),\bm A)\right)\,,
\]
so that, recalling the prescription \eqref{eq:replacements}, adapting the computations in \cite[Lemma 4.2]{GraHaiLewSer-2013-ARMA}, and taking the required traces in the correct order, we split the result as the sum of three terms. We do so in order to separate the zero-temperature part, already studied in \cite{GraHaiLewSer-2013-ARMA} and the magnetic contribution at positive temperature found in \cite{Mor-2026-AHP}. Our goal here is to extract the additional electrostatic term at $T>0$, not present in the mentioned works.
\begin{equation}\label{eq:F2-decomposition-T21-T22-T23}
    \tr\left(\tr_{\C^4}T_2(\bm A,\beta) \right)=\mathcal T_{2,1}(\bm A,\beta)+\mathcal T_{2,2}(\bm A,\beta)+\mathcal T_{2,3}(\bm A,\beta)\,,
\end{equation}
where
\begin{equation}\label{eq:T21}
\mathcal T_{2,1}(\bm A,\beta)\coloneqq-\frac{1}{2\beta\pi^3}\int_{\R^6}\sum^2_{j=0}c_j \sum_{l\in\mathbb Z}\frac{\vert \widehat A(k)\vert^2+\vert \widehat V(k)\vert^2}{(p^2+m^2_j+\omega(l,\beta)^2)^2}\omega(l,\beta)^2\,\diff k\diff p\,,
\end{equation}
and 
\begin{equation}\label{eq:T22}
\mathcal T_{2,2}(\bm A,\beta)=\mathcal T_{2,2}(A,\beta)+\mathcal T_{2,2}(V,\beta)\,,
\end{equation}
with
\begin{equation}\label{eq:T22A}
\mathcal T_{2,2}( A,\beta)\coloneqq\frac{1}{2\beta\pi^3}\int_{\R^6}\sum^2_{j=0}c_j\sum_{l\in\mathbb Z}\frac{\vert \widehat A(k) \vert^2\, \omega(l,\beta)^2}{(p^2+m^2_j+\omega(l,\beta)^2)((p-k)^2+m^2_j+\omega(l,\beta)^2)}\,\diff k\diff p
\end{equation}
and
\begin{equation}\label{eq:T22V}
\mathcal T_{2,2}( V,\beta)= \mathcal{T}^1_{2,2}(V,\beta)+\mathcal{T}^2_{2,2}(V,\beta)\,,
\end{equation}
where, in turn,
\begin{equation}\label{eq:T22V1}
    \mathcal{T}^1_{2,2}(V,\beta)\coloneqq\frac{1}{2\beta\pi^3}\int_{\R^6}\sum^2_{j=0}c_j\sum_{l\in\mathbb Z}\frac{\vert \widehat V(k) \vert^2\omega(l,\beta)^2}{(p^2+m^2_j+\omega(l,\beta)^2)((p-k)^2+m^2_j+\omega(l,\beta)^2)}\,\diff k\diff p \,,
\end{equation}
and 
\begin{equation}\label{eq:T22V2}
\begin{split}
\mathcal{T}^2_{2,2}(V,\beta):&=\frac{3}{2\beta\pi^3}\int_{\R^6}\sum^2_{j=0}c_j\sum_{l\in\mathbb Z}\frac{\vert \widehat V(k) \vert^2\omega(l,\beta)^2}{(p^2+m^2_j+\omega(l,\beta)^2)((p-k)^2+m^2_j+\omega(l,\beta)^2)}\,\diff k\diff p \\
&-\frac{2}{\beta\pi^3}\int_{\R^6}\sum^2_{j=0}c_j\sum_{l\in\mathbb Z}\frac{\vert \widehat V(k) \vert^2\omega(l,\beta)^4}{(p^2+m^2_j+\omega(l,\beta)^2)^2((p-k)^2+m^2_j+\omega(l,\beta)^2)}\,\diff k\diff p\,.
\end{split}
\end{equation}
Finally,
\begin{equation}\label{eq:T23}
\mathcal T_{2,3}( \bm A,\beta)\coloneqq\frac{1}{2\beta\pi^3}\int_{\R^6}\sum^2_{j=0}c_j\sum_{l\in\mathbb Z}\frac{\left(k^2 \vert \widehat A(k) \vert^2+(k^2-4p\cdot k)\vert \widehat V(k)\vert^2\right)\omega(l,\beta)^2}{(p^2+m^2_j+\omega(l,\beta)^2)^2((p-k)^2+m^2_j+\omega(l,\beta)^2)}\,\diff k\diff p\,.
\end{equation}
We emphasise once more that here the property $\operatorname{div} A=0$ is essential to obtain \eqref{eq:F2-decomposition-T21-T22-T23} (see \cite[Proof of Lemma~4.2]{GraHaiLewSer-2013-ARMA}).

\begin{rmk}
Notice that the terms appearing in \eqref{eq:F2-decomposition-T21-T22-T23} are the positive-temperature analogue of those found in \cite{GraHaiLewSer-2013-ARMA} at $T=0$. Indeed, the former are the Riemann sum approximations of the latter, which are recovered by taking the limit as $\beta\to+\infty$, that is, as $T\to 0^+$. To see this, let us just consider \eqref{eq:T21} as an example. It suffices to remark that
\[
\frac{1}{2\beta \pi^3}= \frac{2\pi}{\beta}\frac{1}{4\pi^4}=\frac{1}{4\pi^4} \left(\omega(l,\beta)-\omega(l-1,\beta)\right)\,,
\]
so that 
\[
\mathcal T_{2,1}(\bm A, \beta)\xrightarrow{\beta\to+\infty}-\frac{1}{4\pi^4}\int_{\R^7}\sum^2_{j=0}c_j \frac{\vert \widehat A(k)\vert^2+\vert \widehat V(k)\vert^2}{(p^2+m^2_j+\omega^2)^2}\,\omega^2 \diff \omega \diff k \diff p\,,
\]
which is the $\mathcal T_{2,1}$ term described in \cite[Sec. 4.2]{GraHaiLewSer-2013-ARMA} at zero temperature. A similar argument applies to the other terms in \eqref{eq:F2-decomposition-T21-T22-T23}.
\end{rmk}
\par
Collecting all the above terms except $\mathcal{T}^2_{2,2}(V,\beta)$ and repeating the arguments in \cite{GraHaiLewSer-2013-ARMA,Mor-2026-AHP} yields the first line of \eqref{eq:F2}, namely
\begin{equation}
\begin{split}\label{eq:sum-Txy-M0-MT}
\frac{1}{\beta}&\int_0^\beta\left(\mathcal{T}_{2,1}(\bm A,b)+\mathcal T_{2,2}(A,b)+\mathcal T^1_{2,2}(V,b)+\mathcal T_{2,3}(\bm A,b)\right)\diff b\\
&=\frac{1}{8\pi}\int_{\R^3}\left(M^0(k)+M^T(k,\beta)\right)\left(\vert \widehat B(k)\vert^2-\vert \widehat E(k)\vert^2\right)\, \diff k\,,
\end{split}
\end{equation}
where the functions $M^0(k)$ and $M^T(k,\beta)$ are respectively given in \eqref{eq:Mzero} and \eqref{eq:MT}.
\smallskip

Compared to the cases treated in \cite{GraHaiLewSer-2013-ARMA,Mor-2026-AHP}, we have the additional term $\mathcal{T}^2_{2,2}(V,\beta)$ defined in \eqref{eq:T22V2}. The latter leads to the multiplier $\Gamma(k,\beta)$ in \eqref{eq:F2}, that we analyse in the rest of this subsection. Therefore, we deal with the term involving $\mathcal T^2_{2,2}(V,\beta)$, that we rewrite as
\begin{equation}\label{eq:integral-Gamma-V}
\frac{1}{\beta}\int_0^\beta \mathcal T^2_{2,2}(V,b)\,\diff b=\frac{1}{\beta}\int_0^\beta \int_{\R^3}\Gamma(k,b)\vert \widehat V(k)\vert^2\,\diff k\diff b \,,
\end{equation}
where the multiplier $\Gamma(k,\beta)$ is given by
\begin{equation}\label{eq:Gamma}
    \Gamma(k,\beta)=\frac{1}{2\beta\pi^3}\int_{\R^3}\sum^2_{j=0}c_j\sum_{l\in\mathbb Z}\frac{\left(3\left(p^2+m_j^2\right)-\omega\left(l,\beta\right)^2\right)\omega(l,\beta)^2}{\left(p^2+m^2_j+\omega\left(l,\beta\right)^2\right)^2\left(\left(p-k\right)^2+m^2_j+\omega\left(l,\beta\right)^2\right)}\,\diff p\,.
\end{equation}
In order to study the above integral, we exploit the following identity \cite[Chap. 5]{GreRei-2009-book}
\begin{equation*}
    \frac{1}{a^2 b}=\int_0^1\left(\int_0^\infty s^2 e^{-s\left(ua+\left(1-u\right)b\right)}\diff s\right)u\diff u\,,
\end{equation*}
taking
\[
a=p^2+m^2_j+\omega(l,\beta)^2\,,\qquad b=(p-k)^2+m^2_j+\omega(l,\beta)^2\,.
\]
Therefore,
\begin{align}
    &\notag\Gamma\left(k,\beta\right)=\Gamma_1\left(k,\beta\right)+\Gamma_2\left(k,\beta\right)+\Gamma_3\left(k,\beta\right)\\\label{eq:Gamma-1}
    &\coloneqq \frac{1}{2\beta\pi^3}\int_0^1\diff u\,u\int_0^\infty \diff s\,s^2\sum_{j=0}^2c_je^{-s\left(m_j^2+\left(1-u\right)k^2\right)}\left[\sum_{l\in\Z}\omega\left(l,\beta\right)^2e^{-s\omega\left(l,\beta\right)^2}\right]\int_{\R^3}\diff p\,3p^2e^{-s\left(p^2-2p\cdot k\left(1-u\right)\right)}\\\label{eq:Gamma-2}
    &+ \frac{3}{2\beta\pi^3}\int_0^1\diff u\,u\int_0^\infty \diff s\,s^2\sum_{j=0}^2c_jm_j^2e^{-s\left(m_j^2+\left(1-u\right)k^2\right)}\left[\sum_{l\in\Z}\omega\left(l,\beta\right)^2e^{-s\omega\left(l,\beta\right)^2}\right]\int_{\R^3}\diff p\,e^{-s\left(p^2-2p\cdot k\left(1-u\right)\right)}\\\label{eq:Gamma-3}
    &-\frac{1}{2\beta\pi^3}\int_0^1\diff u\,u\int_0^\infty \diff s\,s^2\sum_{j=0}^2c_je^{-s\left(m_j^2+\left(1-u\right)k^2\right)}\left[\sum_{l\in\Z}\omega\left(l,\beta\right)^4e^{-s\omega\left(l,\beta\right)^2}\right]\int_{\R^3}\diff p\,e^{-s\left(p^2-2p\cdot k\left(1-u\right)\right)}\,.
\end{align}

\subsubsection{Analysis of $\Gamma_1(k,\beta)$}
The sum in square brackets can be expressed in terms of the Jacobi theta function $\theta_2$ as
\[
\sum_{l\in\Z}\omega\left(l,\beta\right)^2e^{-s\omega\left(l,\beta\right)^2}=\sum_{l\in\Z}e^{-s\frac{4\pi^2}{\beta^2}\left(l-\frac{1}{2}\right)^2}\frac{4\pi^2}{\beta^2}\left(k-\frac{1}{2}\right)^2=-\frac{\diff}{\diff s}\theta_2\left(0,\frac{4\pi is}{\beta^2}\right)\,,
\]
where $\theta_2$ is defined by
\begin{equation}\label{eq:jacobi-theta-2}
    \theta_2\left(z,\tau\right)=e^{i\tau \frac{\pi}{4}+i\pi z}\theta\left(z+\frac{\tau}{2},\tau\right)=\sum_{n=-\infty}^{+\infty}e^{i\pi z\left(2n+1\right)}e^{i\pi\tau\left(n+\frac{1}{2}\right)^2}\,,
\end{equation}
with
\[
\theta\left(z,\tau\right)=\sum_{n=-\infty}^{+\infty}e^{2\pi i nz+\pi i n^2 \tau}
\]
(to our knowledge, it does not exist an unambiguous definition in the literature, see \cite[Chap.~16]{AbrSte-1964-book} and \cite{Bel-1961-book} for instance). Moreover, the Gaussian integral in \eqref{eq:Gamma-1} can be computed as follows
\[
\begin{split}
\int_{\R^3}\diff p& \,p^2e^{-s(p^2-2p\cdot k(1-u))}\\
&=e^{sk^2(1-u)^2}\int_{\R^3}\diff p\, p^2e^{-s(p-k(1-u))^2} \\
&\stackrel{q\coloneqq p-k(1-u)}{=}e^{sk^2(1-u)^2}\int_{\R^3}\left(q^2+2q\cdot k\left(1-u\right)+\left(1-u\right)^2k^2\right)e^{-sq^2}\,\diff q \\
&=e^{sk^2\left(1-u\right)^2}\left[\left(1-u\right)^2k^2\int_{\R^3}e^{-sq^2}\diff q+2\left(1-u\right)\int_{\R^3}k\cdot qe^{-sq^2}\diff q+\int_{\R^3}q^2e^{-sq^2}\diff q\right]\,.
\end{split}
\]
Now, the second integral in the last line vanishes due to simmetry reasons. Concerning the first and the third integrals, one can easily show that
\begin{align*}
    \int_{\R^3}e^{-sq^2}\diff q&=\left(\frac{\pi}{s}\right)^{3/2}\,,\\
    \int_{\R^3}q^2e^{-sq^2}\diff q&=\frac{2\pi}{s^2}\,,
\end{align*}
yielding
\begin{multline*}
    \Gamma_1\left(k,\beta\right)=\frac{3}{2\beta\pi^3}\int_0^1\diff u\,u\int_0^\infty\diff s\sum_{j=0}^2c_je^{-s\left(m_j^2+u\left(1-u\right)k^2\right)}\\
    \times\left(-\frac{\diff}{\diff s}\theta_2\left(0,\frac{4\pi is}{\beta^2}\right)\right)\left(\left(1-u\right)^2k^2\pi^{3/2}s^{1/2}+2\pi\right)\,.
\end{multline*}
Integrating by parts with respect to the $s-$variable, one gets
\begin{align*}
    \Gamma_1\left(k,\beta\right)&=\frac{3k^2}{2\beta\pi^{3/2}}\int_0^1\diff u\,u\left(1-u\right)^2\int_0^\infty\diff s\sum_{j=0}^2c_j\theta_2\left(0,\frac{4\pi is}{\beta^2}\right)e^{-s\left(m_j^2+u\left(1-u\right)k^2\right)}\times \\ 
    &\hspace{7cm}\times \left(\frac{1}{2s^{1/2}}-s^{1/2}\left(m_j^2+u\left(1-u\right)k^2\right)\right)\\
    &\quad-\frac{3}{\beta\pi^2}\int_0^1\diff u\,u\int_0^\infty\diff s\sum_{j=0}^2c_j\theta\left(0,\frac{4\pi is}{\beta^2}\right)e^{-s\left(m_j^2+u\left(1-u\right)k^2\right)}\left(m_j^2+u\left(1-u\right)k^2\right)\\
    &\quad-\frac{3}{2\beta\pi^3}\int_0^1\diff u\,u\sum_{j=0}^2c_j\left[e^{-s\left(m_j^2+u\left(1-u\right)k^2\right)}\theta_2\left(0,\frac{4\pi is}{\beta^2}\right)\left(\left(1-u\right)^2k^2\pi^{3/2}s^{1/2}+2\pi\right)\right]_0^\infty\,.
\end{align*}
We will exploit the following properties of $\theta_2$. Using the Poisson summation formula \cite[Formula~6.5]{Bel-1961-book} we get
\begin{equation}\label{eq:jacobi-poisson}
\theta_2\left(0,\frac{4\pi is}{\beta^2}\right)=\sum_{n=-\infty}^{+\infty}e^{-s\frac{4\pi^2}{\beta^2}\left(n-\frac{1}{2}\right)^2}=\left(\frac{1}{\pi s}\right)^{1/2}\frac{\beta}{2}\sum^{+\infty}_{n=-\infty}e^{-\frac{\beta^2}{4s}n^2}e^{-i\pi n}\,.
\end{equation}
On the one hand, from the first identity above, we easily see that
\[
0\leq \theta_2\left(0,\frac{4\pi is}{\beta^2}\right)=\left(\sum_{n=-\infty}^{+\infty}e^{-s\frac{4\pi^2}{\beta^2}\left(n^2-n\right)}\right)\, e^{-s\frac{\pi^2}{\beta^2}}
\leq \left(\sum_{n=-\infty}^{+\infty}e^{-\frac{4\pi^2}{\beta^2}\left(n^2-n\right)}\right)\, e^{-s\frac{\pi^2}{\beta^2}}\,,\qquad s\geq 1\,,
\]
which shows that $\theta_2\left(0,\frac{4\pi is}{\beta^2}\right)\to 0$, as $s\to +\infty$. On the other hand, looking at the last term in \eqref{eq:jacobi-poisson} yields
\[
0\leq \abs{\theta_2\left(0,\frac{4\pi is}{\beta^2}\right)}\leq \left(\frac{1}{\pi s}\right)^{1/2}\frac{\beta}{2}\sum^{+\infty}_{n=-\infty}\abs{e^{-\frac{\beta^2}{4}n^2}} \,,\qquad 0<s<1\,,
\]
proving that $\theta_2\left(0,\frac{4\pi is}{\beta^2}\right)\sim s^{-1/2}$, as $s\to 0$.

Thus, we deduce that the square brackets in the last line below vanish,
\begin{align*}
    &\sum_{j=0}^2c_j\left[e^{-s\left(m_j^2+u\left(1-u\right)k^2\right)}\theta_2\left(0,\frac{4\pi is}{\beta^2}\right)\left(\left(1-u\right)^2k^2\pi^{3/2}s^{1/2}+2\pi\right)\right]_{s=0}^{s=\infty}\\
    &\hspace{1cm}=\sum_{j=0}^2c_j\left[e^{-s\left(m_j^2+u\left(1-u\right)k^2\right)}\theta_2\left(0,\frac{4\pi is}{\beta^2}\right)2\pi\right]_{s=0}^{s=\infty}\\
    &\hspace{1cm}=2\pi\left[\underbrace{\underbrace{\sum_{j=0}^2c_j}_{=0, \mbox{ by \eqref{eq:PV-conditions}}}}
    \theta_2\left(0,\frac{4\pi is}{\beta^2}\right)+\sum_{j=0}^2c_j\underbrace{\theta_2\left(0,\frac{4\pi is}{\beta^2}\right)\left(e^{-s\left(m_j^2+u\left(1-u\right)k^2\right)}-1\right)}_{\bigO{s^{1/2}}\text{ for }s\to 0+}\right]_{s=0}^{s=\infty}=0\,,
\end{align*}
which implies
\begin{align*}
&\Gamma_1\left(k,\beta\right)=\Gamma_{1,1}\left(k,\beta\right)+\Gamma_{1,2}\left(k,\beta\right)\\
    &\coloneqq\frac{3k^2}{2\beta\pi^{3/2}}\int_0^1\diff u\,u\left(1-u\right)^2\int_0^\infty\diff s\sum_{j=0}^2c_j\theta_2\left(0,\frac{4\pi is}{\beta^2}\right)e^{-s\left(m_j^2+u\left(1-u\right)k^2\right)}\left(\frac{1}{2s^{1/2}}-s^{1/2}\left(m_j^2+u\left(1-u\right)k^2\right)\right)\\
    &\quad-\frac{3}{\beta\pi^2}\int_0^1\diff u\,u\int_0^\infty \diff s\sum_{j=0}^2c_j\theta_2\left(0,\frac{4\pi is}{\beta^2}\right)e^{-s\left(m_j^2+u\left(1-u\right)k^2\right)}\left(m_j^2+u\left(1-u\right)k^2\right)\,.
\end{align*}\par
Let us first focus on $\Gamma_{1,1}$. This term can be further rewritten extracting the $n=0$ term in the series in \eqref{eq:jacobi-poisson},
\begin{equation}\label{eq:theta-function-poisson}
    \theta_2\left(0,\frac{4\pi is}{\beta^2}\right)=\left(\frac{1}{\pi s}\right)^{1/2}\frac{\beta}{2}\left[1+\sum_{n\neq 0}e^{-\frac{\beta^2}{4s}n^2}e^{-i\pi n}\right]\,,
\end{equation}
where in the second equality we have used the Poisson summation formula \cite[Formula~6.5]{Bel-1961-book}. Accordingly, we get
\begin{equation}\label{eq:Gamma-11-decomposition}
\Gamma_{1,1}\left(k,\beta\right)=\Gamma^0_{1,1}\left(k\right)+\Gamma^T_{1,1}\left(k,\beta\right)\,,
\end{equation}
where
\begin{equation*}
    \Gamma_{1,1}^0\left(k\right)=\frac{3k^2}{4\pi^2}\int_0^1\diff u\,u\left(1-u\right)^2\int_0^\infty\diff s\sum_{j=0}^2c_je^{-s\left(m_j^2+u\left(1-u\right)k^2\right)}\left(\frac{1}{2s}-\left(m_j^2+u\left(1-u\right)k^2\right)\right)
\end{equation*}
and
\begin{multline*}
    \Gamma_{1,1}^T\left(k,\beta\right)=\frac{3k^2}{4\pi^2}\sum_{n\neq 0}\left(-1\right)^n\int_0^1\diff u\,u\left(1-u\right)^2\int_0^\infty\diff s\sum_{j=0}^2c_je^{-s\left(m_j^2+u\left(1-u\right)k^2\right)-\frac{\beta^2n^2}{4}\frac{1}{s}}\times \\
    \times\left(\frac{1}{2s}-\left(m_j^2+u\left(1-u\right)k^2\right)\right)\,.
\end{multline*}\par
Let us first analyse $\Gamma_{1,1}^0$. Notice that, integrating by parts in the $s$--variable, one has
\begin{multline*}
    \int^\infty_0\,\diff s\, \sum^2_{j=0}c_j e^{-s(m^2_j+u(1-u)k^2)}\,\frac{1}{2s}\\
    =\int^\infty_0\,\diff s\, \sum^2_{j=0}c_j e^{-s(m^2_j+u(1-u)k^2)}(m^2_j+u(1-u)k^2)\,\log s +\left[\sum^2_{j=0}c_j e^{-s(m^2_j+u(1-u)k^2)} \log s\right]^{s=+\infty}_{s\to 0^+}\,,
\end{multline*}
and the boundary terms above vanish, since 
\[
\begin{split}
\left[\sum^2_{j=0}c_j e^{-s(m^2_j+u(1-u)k^2)} \log s\right]^{s=+\infty}_{s\to 0^+}&=-\lim_{s\to 0^+}\sum^2_{j=0}c_j e^{-s(m^2_j+u(1-u)k^2)} \log s \\
&=\lim_{s\to0^+}\underbrace{\sum^2_{j=0}c_j  \log s}_{=0, \mbox{ by \eqref{eq:PV-conditions}}}+ \lim_{s\to0^+} \sum^2_{j=0}c_j \left( e^{-s(m^2_j+u(1-u)k^2)}-1\right) \log s \\
&=0\,.
\end{split}
\]
Then, making the change of variable $\sigma=s(m^2_j+u(1-u)k^2)$ in the remaining integral, a direct computation gives
\[
\int^\infty_0\,\diff s\, \sum^2_{j=0}c_j e^{-s(m^2_j+u(1-u)k^2)}(m^2_j+u(1-u)k^2)\,\log s=-\sum^2_{j=0}c_j\log(m^2_j+u(1-u)k^2)\,.
\]
Moreover, notice that
\[
\int^\infty_0\,\diff s\, \sum^2_{j=0}c_j e^{-s(m^2_j+u(1-u)k^2)}\,(m^2_j+u(1-u)k^2)=\sum^2_{j=0}c_j \left[e^{-s(m^2_j+u(1-u)k^2)}\right]_{s=0}^{s=\infty}=0\,,
\]
using again \eqref{eq:PV-conditions}. Finally, we find
\begin{equation}\label{eq:Gamma-11-0}
\boxed{
    \Gamma^0_{1,1}\left(k\right)=-\frac{3k^2}{8\pi^2}\int^1_0\diff u\,u\left(1-u\right)^2 \sum^2_{j=0}c_j\log(m^2_j+u(1-u)k^2)\,.
    }
\end{equation}\par
We now deal with the term $\Gamma^T_{1,1}$ in \eqref{eq:Gamma-11-decomposition}, starting from the following identity (see \cite[Formulas~(3.324(1))~and~(3.471(9))]{GraRyz-1965-book}) involving modified Bessel functions of second kind $K_\nu$:
\begin{equation}\label{eq:Bessel-integral-formula}
    \int_{0}^\infty x^{\nu-1}\exp{\left(-\alpha\frac{1}{x}-\gamma x\right)}\diff x=2\left(\frac{\alpha}{\gamma}\right)^{\nu/2}K_\nu\left(2\sqrt{\alpha\nu}\right),\quad\mathrm{Re}\beta,\mathrm{Re}\gamma>0\,,
\end{equation}
which allows to rewrite
\[
\int_0^\infty\sum_{j=0}^2 c_j \frac{1}{2s}e^{-s\left(m_j^2+u\left(1-u\right)k^2\right)-\frac{\beta^2 n^2}{4}\frac{1}{s}}\diff s=\sum_{j=0}^2 c_j K_0\left(\beta n\sqrt{m_j^2+u\left(1-u\right)k^2}\right)
\]
and
\begin{multline*}
    \int_0^\infty\sum_{j=0}^2 c_j e^{-s\left(m_j^2+u\left(1-u\right)k^2\right)-\frac{\beta^2 n^2}{4}\frac{1}{s}}\left(m_j^2+u\left(1-u\right)k^2\right)\diff s\\
    =\sum_{j=0}^2c_j \beta n\sqrt{m_j^2+u\left(1-u\right)k^2}K_1\left(\beta n\sqrt{m_j^2+u\left(1-u\right)k^2}\right)\,.
\end{multline*}
Taking into account the parity of Bessel functions, we get
\[
\Gamma^T_{1,1}\left(k,\beta\right)=\frac{3k^2}{4\pi^2}\sum_{n=1}^\infty(-1)^n\int^1_0\diff u\, u\left(1-u\right)^2\sum^2_{j=0}c_j\left[K_0\left(nX_j\right)-n X_j K_1\left(n X_j\right) \right]\,,
\]
where we have set
\begin{equation}\label{eq:simple-Xj}
    X_j=X_j(\beta,u,k):=\beta\sqrt{m^2_j+u(1-u)k^2}\,.
\end{equation}
Finally, the Sommerfeld integral representation \cite[pp.~328]{NikOuv-1983-book} for Bessel functions of second kind,
\begin{equation}\label{eq:Sommerfeld-representation}
    K_\nu\left(x\right)=\int_0^\infty e^{-x\cosh t}\cosh\left(\nu t\right)\diff t,\;x>0\,,
\end{equation}
together with the identities
\[
\sum_{n=1}^\infty\left(-e^{-X_j\cosh t}\right)^n= -\frac{e^{-X_j\cosh t}}{1+e^{-X_j\cosh t}}
\]
and
\begin{align*}
    \sum_{n=1}^\infty n X_j\cosh t\left(-e^{-X_j\cosh t}\right)^n&=X_j\cosh t\,e^{-X_j\cosh t}\frac{\diff}{\diff y}\frac{y}{1+y}\bigg|_{y=e^{-X_j\cosh t}}\\
    &=X_j\cosh t\frac{e^{-X_j\cosh t}}{\left(1+e^{-X_j\cosh t}\right)^2}\,,
\end{align*}
yields
\begin{equation}\label{eq:Gamma-11-T}
    \boxed{
    \Gamma^T_{1,1}\left(k,\beta\right)=-\frac{3k^2}{4\pi^2}\int^1_0\diff u\, u\left(1-u\right)^2 \int^\infty_0\diff t \sum^2_{j=0}c_j \frac{e^{- X_j\cosh t}\left( 1+ X_j \cosh t+e^{-X_j \cosh t}\right)}{\left(1+e^{X_j \cosh t}\right)^2}\,.
    }
\end{equation}\par
Let us now focus on $\Gamma_{1,2}$. As for $\Gamma_{1,1}$, this term can be rewritten by means of the representation formula \eqref{eq:theta-function-poisson} for $\theta_2$, yielding
\begin{equation}\label{eq:Gamma-12-decomposition}
    \Gamma_{1,2}\left(k,\beta\right)=\Gamma_{1,2}^0\left(k\right)+\Gamma_{1,2}^T\left(k,\beta\right)\,,
\end{equation}
where
\begin{equation*}
    \Gamma_{1,2}^0\left(k\right)=-\frac{3}{2\pi^{5/2}}\int_0^1\diff u\,u\int_0^\infty\diff s\sum_{j=0}^2c_je^{-s\left(m_j^2+u\left(1-u\right)k^2\right)}\frac{1}{s^{1/2}}\left(m_j^2+u\left(1-u\right)k^2\right)
\end{equation*}
and
\begin{equation*}
    \Gamma_{1,2}^T\left(k,\beta\right)=-\frac{3}{2\pi^{5/2}}\sum_{n\neq 0}\left(-1\right)^n\int_0^1\diff u\,u\int_0^\infty\diff s\sum_{j=0}^2c_je^{-s\left(m_j^2+u\left(1-u\right)k^2\right)-\frac{\beta^2n^2}{4}\frac{1}{s}}\frac{1}{s^{1/2}}\left(m_j^2+u\left(1-u\right)k^2\right)\,.
\end{equation*}\par
Let us first analyse $\Gamma_{1,2}^0$. Notice that, integrating by parts in the $s$--variable, one has
\begin{align*}
    \int_0^\infty\diff s\sum_{j=0}^2c_je^{-s\left(m_j^2+u\left(1-u\right)k^2\right)}&\frac{1}{s^{1/2}}\left(m_j^2+u\left(1-u\right)k^2\right)
    \\
    &=2\int_0^\infty\diff s\sum_{j=0}^2c_je^{-s\left(m_j^2+u\left(1-u\right)k^2\right)}s^{1/2}\left(m_j^2+u\left(1-u\right)k^2\right)^2\\
    &=\sqrt{\pi}\sum_{j=0}^2c_j\left(m_j^2+u\left(1-u\right)k^2\right)^{1/2}\,,
\end{align*}
where in the last equality we have computed the Gaussian integral, using the change of variable $\sigma=s\left(m_j^2+u\left(1-u\right)k^2\right)$. This leads to
\begin{equation}\label{eq:Gamma-12-0}
    \boxed{\Gamma_{1,2}^0\left(k\right)=-\frac{3}{2\pi^2}\int_0^1\diff u\,u\sum_{j=0}^2c_j\left(m_j^2+u\left(1-u\right)k^2\right)^{1/2}\,.
    }
\end{equation}\par
We now deal with the term $\Gamma_{1,2}^T$ in \eqref{eq:Gamma-12-decomposition}. Using the integral identity in \eqref{eq:Bessel-integral-formula} involving modified Bessel functions of second kind, with $\nu=1/2$, $\gamma=\left(m_j^2+u\left(1-u\right)k^2\right)$ and $\alpha=\beta^2n^2/4$, one has
\begin{multline*}
    \int_0^\infty\diff s\sum_{j=0}^2c_je^{-s\left(m_j^2+u\left(1-u\right)k^2\right)-\frac{\beta^2n^2}{4}\frac{1}{s}}s^{-1/2}\left(m_j^2+u\left(1-u\right)k^2\right)\\
    =\sum_{j=0}^2c_j\sqrt{2\beta n}\left(m_j^2+u\left(1-u\right)k^2\right)^{3/2}K_{1/2}\left(\beta n\sqrt{m_j^2+u\left(1-u\right)k^2}\right)\,.
\end{multline*}
Taking into account the parity of Bessel functions and using the Sommerfeld integral representation formula in \eqref{eq:Sommerfeld-representation}, we finally get
\begin{equation}\label{eq:Gamma-12-T}
\boxed{
    \Gamma_{1,2}^T\left(k,\beta\right)=-\frac{3}{\pi^{5/2}}\int_0^1\diff u\,u\int_0^\infty\diff t\sum_{j=0}^2c_j\sum_{n=1}^\infty\left(-1\right)^n\left[\sqrt{2\beta n}\left(m_j^2+u\left(1-u\right)k^2\right)^{3/2}\cosh\left(\frac{t}{2}\right)e^{-nX_j\cosh t}\right]\,.
    }
\end{equation}
where $X_j$ is defined in \eqref{eq:simple-Xj}.

\subsubsection{Analysis of $\Gamma_2(k,\beta)$}
Expressing as before the sum in square brackets in terms of the Jacobi theta function $\theta_2$ defined in \eqref{eq:jacobi-theta-2} and computing the Gaussian integral in \eqref{eq:Gamma-2} yields
\begin{equation*}
    \Gamma_2\left(k,\beta\right)=\frac{3}{2\beta\pi^3}\int_0^1\diff u\,u\int_0^\infty\diff s\,s^2\sum_{j=0}^2c_jm_j^2e^{-s\left(m_j^2+u\left(1-u\right)k^2\right)}\left(-\frac{\diff}{\diff s}\theta_2\left(0,\frac{4\pi is}{\beta^2}\right)\right)\left(\frac{\pi}{s}\right)^{3/2}\,.
\end{equation*}
Integrating by parts in the $s$--variable and then using again the formula in \eqref{eq:theta-function-poisson},
one gets
\begin{align*}
    &\Gamma_2\left(k,\beta\right)=\Gamma_2^0\left(k\right)+\Gamma_2^T\left(k,\beta\right)\\
    &\coloneqq \frac{3}{4\pi^2}\int_0^1\diff u\, u\int_0^\infty\diff s\sum_{j=0}^2c_jm_j^2e^{-s\left(m_j^2+u\left(1-u\right)k^2\right)}\left(\frac{1}{2s}-\left(m_j^2+u\left(1-u\right)k^2\right)\right)\\
    &\quad+\frac{3}{4\pi^2}\sum_{n\neq 0}\left(-1\right)^n\int_0^1 \diff u\,u\int_o^\infty\diff s\sum_{j=0}^2c_jm_j^2e^{-s\left(m_j^2+u\left(1-u\right)k^2\right)-\frac{\beta^2 n^2}{4}\frac{1}{s}}\left(\frac{1}{2s}-\left(m_j^2+u\left(1-u\right)k^2\right)\right)\,.
\end{align*}\par
Let us now focus on $\Gamma_2^0$. Since
\begin{equation*}
    \int_0^\infty \diff s\sum_{j=0}^2c_jm_j^2e^{-s\left(m_j^2+u\left(1-u\right)k^2\right)}s^{-1}=-\sum_{j=0}^2c_jm_j^2\log\left(m_j^2+u\left(1-u\right)k^2\right)
\end{equation*}
and
\begin{equation*}
    -\int_0^\infty\diff s\sum_{j=0}^2c_jm_j^2e^{-s\left(m_j^2+u\left(1-u\right)k^2\right)}\left(m_j^2+u\left(1-u\right)k^2\right)=\sum_{j=0}^2c_jm_j^2\left[e^{-s\left(m_j^2+u\left(1-y\right)k^2\right)}\right]_{s=0}^{s=\infty}=0\,,
\end{equation*}
where in the latter we have used the Pauli-Villars conditions \eqref{eq:PV-conditions}, one has
\begin{equation}\label{eq:Gamma-2-0}
    \boxed{
    \Gamma_2^0\left(k\right)=-\frac{3}{8\pi^2}\int_0^1\diff u\,u\sum_{j=0}^2c_jm_j^2\log\left(m_j^2+u\left(1-u\right)k^2\right)\,.
    }
\end{equation}\par
Let us now move to $\Gamma_2^T$. Here, we use
\begin{equation*}
    \int_0^\infty\diff s\sum_{j=0}^2c_jm_j^2e^{-s\left(m_j^2+u\left(1-u\right)k^2\right)-\frac{\beta^2n^2}{4}\frac{1}{s}}\frac{1}{2s}=\sum_{j=0}^2c_jm_j^2K_0\left(nX_j\right)
\end{equation*}
and
\begin{equation*}
    \int_0^\infty\diff s\sum_{j=0}^2c_jm_j^2e^{-s\left(m_j^2+u\left(1-u\right)k^2\right)-\frac{\beta^2n^2}{4}\frac{1}{s}}\left(m_j^2+u\left(1-u\right)k^2\right)=\sum_{j=0}^2c_jm_j^2nX_jK_1\left(nX_j\right),
\end{equation*}
where $X_j$ is given in \eqref{eq:simple-Xj}, leading to
\begin{equation*}
    \Gamma_2^T\left(k,\beta\right)=\frac{3}{2\pi^2}\sum_{n=1}^\infty\left(-1\right)^n\int_0^1\diff u\,u\sum_{j=0}^2c_jm_j^2\left[K_0\left(nX_j\right)-nX_jK_1\left(nX_j\right)\right]\,,
\end{equation*}
due to the parity of Bessel functions. Now, by exploiting the Sommerfeld representation in \eqref{eq:Sommerfeld-representation}, we can rewrite $\Gamma_2^T$ as
\begin{equation*}
    \Gamma_2^T\left(k,\beta\right)=\frac{3}{2\pi^2}\int_0^1\diff u\,u\int_0^\infty\diff t\sum_{j=0}^2c_jm_j^2\sum_{n=1}^\infty\left(-1\right)^n\left[e^{-nX_j\cosh t}-nX_j\cosh t\,e^{-nX_j\cosh t}\right]\,,
\end{equation*}
which finally leads to
\begin{equation}\label{eq:Gamma-2-T}
    \boxed{
    \Gamma_2^T\left(k,\beta\right)=-\frac{3}{2\pi^2}\int_0^1\diff u\,u\int_0^\infty\diff t\sum_{j=0}^2c_jm_j^2\frac{e^{-X_j\cosh t}\left(1+X_j\cosh t+e^{-X_j\cosh t}\right)}{\left(1+e^{-X_j\cosh t}\right)^2}\,.
    }
\end{equation}

\subsubsection{Analysis of $\Gamma_3(k,\beta)$}
Here, the sum in square brackets can be expressed in terms of the Jacobi theta function $\theta_2$ in a slightly different way as follows
\begin{equation*}
    \sum_{l\in\Z}\omega\left(l,\beta\right)^2e^{-s\omega\left(l,\beta\right)^2}=-\frac{\diff^2}{\diff s^2}\theta_2\left(0,\frac{4\pi is}{\beta^2}\right)\,.
\end{equation*}
This, together with the value of the Gaussian integral in \eqref{eq:Gamma-3}, leads to
\begin{equation*}
    \Gamma_3\left(k,\beta\right)=-\frac{1}{2\beta\pi^3}\int_0^1\diff u\,u\int_0^\infty\diff s\,s^2\sum_{j=0}^2c_je^{-s\left(m_j^2+u\left(1-u\right)k^2\right)}\left(-\frac{\diff^2}{\diff s^2}\theta_2\left(0,\frac{4\pi i s}{\beta^2}\right)\right)\left(\frac{\pi}{s}\right)^{3/2}\,.
\end{equation*}
Integrating twice by parts and then using once again \eqref{eq:theta-function-poisson},
one gets
\begin{equation*}
\Gamma_3\left(k,\beta\right)=\Gamma_3^0\left(k\right)+\Gamma_3^T\left(k,\beta\right)\,,
\end{equation*}
where
\begin{multline*}
    \Gamma_3^0\left(k\right)\coloneqq\frac{1}{4\pi^2}\int_0^1\diff u\,u\int_0^\infty\diff s\sum_{j=0}^2c_je^{-s\left(m_j^2+u\left(1-u\right)k^2\right)}\times \\
    \times\left(-\frac{1}{4s^2}-\frac{1}{s}\left(m_j^2+u\left(1-u\right)k^2\right)+\left(m_j^2+u\left(1-u\right)k^2\right)\right) 
\end{multline*}
and
\begin{multline*}
    \Gamma_3^T\left(k,\beta\right)\coloneqq\frac{1}{4\pi^2}\sum_{n\neq 0}\left(-1\right)^n\int_0^1\diff u\,u\int_0^\infty\diff s\sum_{j=0}^2c_je^{-s\left(m_j^2+u\left(1-u\right)k^2\right)-\frac{\beta^2n^2}{4}\frac{1}{s}}\times \\
    \times\left(-\frac{1}{4s^2}-\frac{1}{s}\left(m_j^2+u\left(1-u\right)k^2\right)+\left(m_j^2+u\left(1-u\right)k^2\right)\right).
\end{multline*}\par
Let us now focus on $\Gamma_3^0$. Similarly to the study of $\Gamma_2$, we have
\begin{equation}\label{eq:Gamma-3-0-aux1}
    \int_0^\infty\diff s\sum_{j=0}^2c_je^{-s\left(m_j^2+u\left(1-u\right)k^2\right)}\left(m_j^2+u\left(1-u\right)k^2\right)=-\sum_{j=0}^2c_j\left[e^{-s\left(m_j^2+u\left(1-u\right)k^2\right)}\right]_{s=0}^{s=\infty}=0\,,
\end{equation}
using the Pauli-Villars conditions \eqref{eq:PV-conditions} for $s=0$, and
\begin{multline}
    -\int_0^\infty\diff s\sum_{j=0}^2c_je^{-s\left(m_j^2+u\left(1-u\right)k^2\right)}s^{-1}\left(m_j^2+u\left(1-u\right)k^2\right)\\\label{eq:Gamma-3-0-aux2}
    =\sum_{j=0}^2c_j\log\left(m_j^2+u\left(1-u\right)k^2\right)\left(m_j^2+u\left(1-u\right)k^2\right)\,.
\end{multline}
In addition, here we need to compute
\begin{align*}
    \int_0^\infty\diff s\sum_{j=0}^2c_j&e^{-s\left(m_j^2+u\left(1-u\right)k^2\right)}\frac{1}{s^2}\\
    &=-\int_0^\infty\diff s\sum_{j=0}^2c_je^{-s\left(m_j^2+u\left(1-u\right)k^2\right)}\frac{\diff}{\diff s}\frac{1}{s}\\
    &=-\lim_{\varepsilon\to 0+}\int_{\varepsilon}^\infty\diff s\left[\underbrace{\sum_{j=0}^2c_j}_{=0, \mbox{ by \eqref{eq:PV-conditions}}}\frac{\diff}{\diff s}\frac{1}{s}+\sum_{j=0}^2c_j\frac{\diff}{\diff s}\frac{1}{s}\left(e^{-s\left(m_j^2+u\left(1-u\right)k^2\right)}-1\right)\right]\,,
\end{align*}
which, integrating by parts, implies
\begin{multline*}
    \int_0^\infty\diff s\sum_{j=0}^2c_je^{-s\left(m_j^2+u\left(1-u\right)k^2\right)}\frac{1}{s^2}=-\lim_{\varepsilon\to 0+}\left[\sum_{j=0}^2c_j\underbrace{\frac{e^{-s\left(m_j^2+u\left(1-u\right)k^2\right)}-1}{s}}_{-\left(m_j^2+u\left(1-u\right)k^2\right)+\smallo{1}\text{ for }s\to 0+}\right]_{s=\varepsilon}^{s=\infty}\\
    -\int_0^\infty\diff s\sum_{j=0}^2c_j\frac{1}{s}e^{-s\left(m_j^2+u\left(1-u\right)k^2\right)}\left(m_j^2+u\left(1-u\right)k^2\right)\,.
\end{multline*}
Finally,
\begin{align}
    &\notag\int_0^\infty\diff s\sum_{j=0}^2c_je^{-s\left(m_j^2+u\left(1-u\right)k^2\right)}\frac{1}{s^2}\\
    &\quad\notag=\underbrace{\sum_{j=0}^2c_j\left(m_j^2+u\left(1-u\right)k^2\right)}_{=0, \mbox{ by \eqref{eq:PV-conditions}}}-\int_0^\infty\diff s\sum_{j=0}^2c_je^{-s\left(m_j^2+u\left(1-u\right)k^2\right)}s^{-1}\left(m_j^2+u\left(1-u\right)k^2\right)\\
    &\quad\label{eq:Gamma-3-0-aux3}=\sum_{j=0}^2c_j\log\left(m_j^2+u\left(1-u\right)k^2\right)\left(m_j^2+u\left(1-u\right)k^2\right).
\end{align}
Now, \eqref{eq:Gamma-3-0-aux1}, \eqref{eq:Gamma-3-0-aux2} and \eqref{eq:Gamma-3-0-aux3} allow to rewrite
\begin{equation}\label{eq:Gamma-3-0}
    \boxed{
    \Gamma_3^0\left(k\right)=\frac{5}{16\pi^2}\int_0^1\diff u\,u\sum_{j=0}^2c_j\log\left(m_j^2+u\left(1-u\right)k^2\right)\left(m_j^2+u\left(1-u\right)k^2\right)\,.
    }
\end{equation}\par
Let us now move to $\Gamma_3^T$. Similarly to the study of $\Gamma_2^T$, we have
\begin{multline}\label{eq:Gamma-3-T-aux1}
    \int_0^\infty\diff s\sum_{j=0}^2c_je^{-s\left(m_j^2+u\left(1-u\right)k^2\right)-\frac{\beta^2n^2}{4}\frac{1}{s}}\frac{1}{s}\left(m_j^2+u\left(1-u\right)k^2\right)\\
    =\sum_{j=0}^2c_j\left(m_j^2+u\left(1-u\right)k^2\right)K_0\left(nX_j\right)
\end{multline}
and
\begin{equation}\label{eq:Gamma-3-T-aux2}
    \int_0^\infty\diff s\sum_{j=0}^2c_je^{-s\left(m_j^2+u\left(1-u\right)k^2\right)-\frac{\beta^2n^2}{4}\frac{1}{s}}\left(m_j^2+u\left(1-u\right)k^2\right)=\sum_{j=0}^2c_jnX_jK_1\left(nX_j\right)\,,
\end{equation}
where $X_j$ is defined in \eqref{eq:simple-Xj}. In addition, here we need to compute
\begin{equation*}
    \frac{1}{4}\int_0^\infty\diff s\sum_{j=0}^2c_je^{-s\left(m_j^2+u\left(1-u\right)k^2\right)-\frac{\beta^2n^2}{4}\frac{1}{s}}\frac{1}{s^2}\,.
\end{equation*}
To do so, we exploit the following identity \cite[Formula~8.432(7)]{GraRyz-1965-book}
\begin{equation*}
    K_\nu\left(xz\right)=\frac{z^\nu}{2}\int_0^\infty e^{-\frac{x}{2}\left(t+\frac{z^2}{t}\right)}t^{-\nu-1}\diff t\,,
\end{equation*}
taking
\begin{equation*}
    t=s,\quad \nu=1,\quad \frac{x}{2}=m_j^2+u\left(1-u\right)k^2\quad\text{and}\quad\frac{x}{2}z^2=\frac{\beta^2n^2}{4}\iff z=\frac{\beta n}{4\left(m_j^2+u\left(1-u\right)k^2\right)^{1/2}}.
\end{equation*}
Then,
\begin{align}
    &\notag\frac{1}{4}\int_0^\infty\diff s\sum_{j=0}^2c_je^{-s\left(m_j^2+u\left(1-u\right)k^2\right)-\frac{\beta^2n^2}{4}\frac{1}{s}}\frac{1}{s^2}\\
    &\notag\quad=\sum_{j=0}^2c_j\frac{2\left(m_j^2+u\left(1-u\right)k^2\right)^{1/2}}{\beta n}K_1\left(\frac{1}{2}\beta n\left(m_j^2+u\left(1-u\right)k^2\right)^{1/2}\right)\\
    &\label{eq:Gamma-3-T-aux3}\quad=\sum_{j=0}^2c_j\frac{2X_j}{\beta^2 n}K_1\left(\frac{n}{2}X_j\right)\,.
\end{align}\par
Now, \eqref{eq:Gamma-3-T-aux1}, \eqref{eq:Gamma-3-T-aux2} and \eqref{eq:Gamma-3-T-aux3} allow to rewrite $\Gamma_3^T$ as
\begin{empheq}[box=\fbox]{align}   \label{eq:Gamma-3-T}
\begin{split}
    \Gamma_3^T\left(k,\beta\right)=\frac{1}{2\pi^2}\int_0^1\diff u\,u\int_0^\infty &\diff t\sum_{j=0}^2c_j\sum_{n=1}^\infty\left(-1\right)^n\bigg[-\frac{2X_j}{\beta^2 n}e^{-\frac{n}{2}X_j\cosh t}\cosh t \\ &-2\left(m_j^2+u\left(1-u\right)k^2\right)e^{-nX_j\cosh t}
    +nX_je^{-nX_j\cosh t}\cosh t\bigg]\,,
    \end{split}
\end{empheq}
using also the parity of Bessel functions.\par
\medskip
 We are now ready to state the following lemma which summarises the previous calculations and complete the analysis of the second-order term $\mathcal{F}_2\left(\bm{F},\beta\right)$.
\begin{lem}\label{lem:F2_proof}
    For $\bm A\in L^2(\R^3,\R^4) \cap \hd$, we have
    \begin{multline}\label{eq:explicit-F2-lemma}
        \mathcal{F}_2\left(\bm{F},\beta\right)=\frac{1}{8\pi}\int_{\R^3}\left(M^0\left(k\right)+M^T\left(k,\beta\right)\right)\left(\abs{\widehat{B}\left(k\right)}^2-\abs{\widehat{E}\left(k\right)}^2\right)\diff k\\
        +\frac{1}{\beta}\int_0^\beta\int_{\R^3}\frac{\Gamma(k,\beta)}{\vert k\vert^2}\abs{\widehat{E}\left(k\right)}^2\diff k\diff b\,,
    \end{multline}
    where $M^0$ and $M^T$ are respectively defined in \eqref{eq:Mzero} and \eqref{eq:MT}, and 
    $\Gamma \in  C^0(\R^3\times(0,+\infty))$ is an explicit function of $(k,\beta)$, such that, for any $\beta\in\intoo{0}{+\infty}$,
    \begin{equation*}
        \frac{\Gamma(\cdot,\beta)}{\vert \cdot\vert^2}\in L^1_{\mathrm{loc}}(\R^3)\,,\quad \mbox{and}\quad \frac{\Gamma(\cdot,\beta)}{\vert \cdot\vert^2}\in L^\infty(\{\vert k\vert\geq \delta\})\,,\quad\forall \delta>0\,.
    \end{equation*}
\end{lem}
\begin{proof}
    From equations \eqref{eq:F2-decomposition-T21-T22-T23}-\eqref{eq:T23}, we have
    \begin{align*}
        \mathcal{F}_2\left(\bm{F},\beta\right)&=\frac{1}{\beta}\int_0^\beta\tr\left(\tr_{\C^4}T_2\left(\bm{A},b\right)\right)\diff b\\
        &=\frac{1}{\beta}\int_0^\beta\left(\mathcal{T}_{2,1}(\bm A,b)+\mathcal T_{2,2}(A,b)+\mathcal T^1_{2,2}(V,b)+\mathcal T_{2,3}(\bm A,b)\right)\diff b+\frac{1}{\beta}\int_0^\beta \mathcal T^2_{2,2}(V,\beta)\,\diff b\,.
    \end{align*}
    The first integral in the last line makes appear the first term in \eqref{eq:explicit-F2-lemma} thanks to \eqref{eq:sum-Txy-M0-MT},
    \begin{multline*}
        \frac{1}{\beta}\int_0^\beta\left(\mathcal{T}_{2,1}(\bm A,b)+\mathcal T_{2,2}(A,b)+\mathcal T^1_{2,2}(V,b)+\mathcal T_{2,3}(\bm A,b)\right)\diff b\\
        =\frac{1}{8\pi}\int_{\R^3}\left(M^0(k)+M^T(k,\beta)\right)\left(\vert \widehat B(k)\vert^2-\vert \widehat E(k)\vert^2\right)\, \diff k\,.
    \end{multline*}
    On the other hand, \eqref{eq:integral-Gamma-V} allows to rewrite the second integral as
    \begin{align}
        \notag\frac{1}{\beta}\int_0^\beta \mathcal T^2_{2,2}(V,b)\,\diff b&=\frac{1}{\beta}\int_0^\beta \int_{\R^3}\Gamma(k,b)\vert \widehat V(k)\vert^2\,\diff k\diff b \\
        &\label{eq:Gamma_field}=\frac{1}{\beta}\int_0^\beta \int_{\R^3}\frac{\Gamma(k,b)}{\vert k\vert^2}\vert \widehat E(k)\vert^2\,\diff k\diff b\,,
    \end{align}
    where $E=-\nabla V$. Here
    \begin{equation}\label{eq:Gamma_sum}
        \begin{split}
            \Gamma\left(k,\beta\right)&=\Gamma_1\left(k,\beta\right)+\Gamma_2\left(k,\beta\right)+\Gamma_3\left(k,\beta\right)\\
        &=\Gamma_{1,1}^0\left(k\right)+\Gamma_{1,1}^T\left(k,\beta\right)+\Gamma_{1,2}^0\left(k\right)+\Gamma_{1,2}^T\left(k,\beta\right)\\
        &\quad+\Gamma_2^0\left(k\right)+\Gamma_2^T\left(k,\beta\right)\\
        &\quad+\Gamma_3^0\left(k\right)+\Gamma_3^T\left(k,\beta\right)\,,
        \end{split}
    \end{equation}
    with the above addends being defined respectively in \eqref{eq:Gamma-11-0}, \eqref{eq:Gamma-11-T}, \eqref{eq:Gamma-12-0}, \eqref{eq:Gamma-12-T}, \eqref{eq:Gamma-2-0}, \eqref{eq:Gamma-2-T}, \eqref{eq:Gamma-3-0} and \eqref{eq:Gamma-3-T}.
    The final part of the proof is devoted to the study of the properties of $\Gamma(\cdot,\beta)$ and of the quotient $\Gamma(\cdot,\beta)/\vert \cdot\vert^2$ through the analysis of the various terms appearing in \eqref{eq:Gamma_sum}. For the sake of brevity, we only sketch parts of the argument.\par

    First of all, it is not hard to see that the function $\Gamma\in C^0(\R^3\times(0,+\infty))$ is continuous in both variables, using the dominated convergence theorem for the $(k,\beta)$--dependent integrals appearing in \eqref{eq:Gamma_sum}. Notice that $\Gamma$ can be written as the sum of several terms, as in \eqref{eq:Gamma_sum}. However, we only deal with $\Gamma^T_{1,2}(k,\beta)$ in \eqref{eq:Gamma-12-T}, as the other terms can be treated adapting the same argument, or even in a simpler way. Up to a constant, we have
    \begin{equation}\label{eq:ex_dom_conv}
    \Gamma^T_{1,2}(k,\beta)=\int^1_0 \diff u\,u \int^\infty_0\diff t \sum^2_{j=0}c_j \sum^{\infty}_{n=1}(-1)^n \left[\sqrt{2\beta n}(m^2_j+u(1-u)k^2)^{3/2}\cosh(t/2)e^{-nX_j \cosh t} \right]\,,
    \end{equation}
    where we recall that $X_j=\beta\sqrt{m^2_j+u(1-u)k^2}$ and $m_2>m_1>m_0$. We want to prove continuity at any fixed $(\overline k,\overline\beta)\in\R^3\times (0,+\infty)$, and then we take a sequence $\R^3\times(0,+\infty)\supseteq(k_n,\beta_n)\to(\overline k,\overline\beta)$. We can estimate the integrand in \eqref{eq:ex_dom_conv} as follows, assuming $(k,\beta)$ varies in a compact set, so that $\vert k\vert\leq M$ and $ \beta_\mathrm{min}<\vert\beta\vert\leq \beta_\mathrm{max}$, for some constants $M>0$, $\beta_\mathrm{max}>\beta_\mathrm{min}>0$:
    \begin{multline*}
        \left\vert\sum^2_{j=0}\sum^\infty_{n=1} u c_j (-1)^n \left[\sqrt{2\beta n}(m^2_j+u(1-u)k^2)^{3/2}\cosh(t/2)e^{-nX_j \cosh t} \right] \right\vert \\
        \leq\underbrace{\left(\sum^\infty_{n=1}\sqrt{\beta_\mathrm{max} n}e^{-(n \beta_\mathrm{min} m_0)/2}\right)}_{<+\infty}(m^2_2+M^2)^{3/2}\cosh(t/2)e^{-( \beta_\mathrm{min} m_0 \cosh t)/2} \in L^1\left(\left(0,1\right)\times \left(0,+\infty\right); \diff u\diff t\right)\,,
    \end{multline*}
    uniformly in $(k,\beta)\in\{\vert k\vert\leq M\}\times\{\beta_\mathrm{min}<\vert \beta\vert<\beta_\mathrm{max}\}$. Here, we used that $u\in[0,1]$ and $X_j=\beta\sqrt{m^2_j+u(1-u)k^2}\geq \beta m_0$, for $j=0,1,2$.
    Then, by Lebesgue's dominated convergence theorem, we can pass to the limit in \eqref{eq:ex_dom_conv} and get
    \[
    \lim_{n\to+\infty}\Gamma^T_{1,2}(k_n,\beta_n)=\Gamma^T_{1,2}(\overline k,\overline\beta)\,.
    \]\par
    Recall that we are assuming $E \in L^1(\R^3)\cap L^2(\R^3)$, which implies that $\widehat E\in C^0(\R^3)\cap L^2(\R^3)$. This shows that the integral in \eqref{eq:Gamma_field} is locally convergent, as $\Gamma(\cdot,\beta)/\vert \cdot\vert^2\in L^1_{\mathrm{loc}}(\R^3)$.\par
    It remains to study the behaviour of this function at infinity. Once again, we provide a detailed analysis only for one term in \eqref{eq:Gamma_sum}, as the remaining terms can be treated similarly, or even by simpler arguments. Therefore, we focus on
    \[
    \begin{split}
    \frac{\Gamma^0_3(k)}{\vert k\vert^2}&=\frac{5}{16\pi^2}\int_0^1\diff u\,u\sum_{j=0}^2c_j\log\left(m_j^2+u\left(1-u\right)k^2\right)\left(\frac{m_j^2}{\vert k\vert^2}+u\left(1-u\right)\right) \\
    &\quad+\frac{5}{16\pi^2}\int_0^1\diff u\,u\sum_{j=0}^2c_j\log\left(m_j^2+u\left(1-u\right)k^2\right)\frac{m_j^2}{\vert k\vert^2} \\
    &\quad+\frac{5}{16\pi^2}\int_0^1\diff u\,u\sum_{j=0}^2c_j\log\left(m_j^2+u\left(1-u\right)k^2\right)\times u\left(1-u\right)\,,
    \end{split}
    \]
    where the last integral has an apparent logarithmic divergence for high momenta, $\vert k\vert \to\infty$. However, it is actually cancelled thanks to the Pauli-Villars conditions \eqref{eq:PV-conditions}. To see this, let us rewrite the above sum as
    \[
    \begin{split}
    \sum_{j=0}^2c_j\log\left(m_j^2+u\left(1-u\right)\vert k\vert^2\right)&=\sum_{j=0}^2c_j\log\left(\left(\frac{m_j^2}{\vert k\vert^2}+u\left(1-u\right)\right)\vert k\vert^2\right)\\
    &=\sum_{j=0}^2c_j\log\left(\frac{m_j^2}{\vert k\vert^2}+u\left(1-u\right)\right)+\underbrace{\sum^2_{j=0}c_j}_{=0} \log(\vert k\vert^2)\\
    &=\sum_{j=0}^2c_j\log\left(\frac{m_j^2}{\vert k\vert^2}+u\left(1-u\right)\right)\,.
    \end{split}
    \]
    This shows that 
    \[
    \frac{\Gamma(\cdot,\beta)}{\vert \cdot\vert^2}\in L^\infty(\{\vert k\vert\geq\delta\})\,,\qquad \forall\delta>0\,,
    \]
    and then the $k$--integral in \eqref{eq:Gamma_field} is also convergent at infinity, since $\widehat E\in L^2(\R^3)$.
\end{proof}

\subsection{End of the proof of \cref{thm:main}}
Item $i)$ of \cref{thm:main} has been proved in Section \ref{sec:proof_item_i)}, while the validity of formula \eqref{eq:FPV-decomposition} is obtained in \eqref{eq:PVs} and Lemma~\ref{lem:F2_proof}. The remainder estimates \eqref{eq:R-bound} can be found in Lemma \ref{lem:R_est}.
The validity of the identity \eqref{eq:F2} in item $iii)$ of \cref{thm:main} has been proved in \cref{sec:the-quadratic-term}. Therefore, we only need to check whether the integrals in \eqref{eq:F2} can be extended from $\bm{A}\in L^1(\R^3,\R^4)\cap\hd$ and $E=-\nabla V\in L^1(\R^3,\R^3)$ to $\bm{A}\in X$ by density. We only need to deal with the term involving the multiplier $\Gamma(k,\beta)$, for which we have an estimate of the form
\[
\begin{split}
\frac{1}{\beta}&\int_0^\beta\int_{\R^3}\frac{\Gamma(k,b)}{\vert k\vert^2}\abs{\widehat{E}\left(k\right)}^2\diff k\diff b\\ 
&=\frac{1}{\beta}\int_0^\beta\int_{\{\vert k\vert\leq \delta\}}\frac{\Gamma(k,b)}{\vert k\vert^2}\abs{\widehat{E}\left(k\right)}^2\diff k\diff b +\frac{1}{\beta}\int_0^\beta\int_{\{\vert k\vert\geq \delta\}}\frac{\Gamma(k,b)}{\vert k\vert^2}\abs{\widehat{E}\left(k\right)}^2\diff k\diff b \\
&\leq\Vert \widehat E\Vert^2_{L^\infty}\frac{1}{\beta}\int_0^\beta\int_{\{\vert k\vert\leq \delta\}}\frac{\Gamma(k,b)}{\vert k\vert^2}\diff k\diff b+\Vert \widehat E\Vert^2_{L^2}\left\Vert \frac{\Gamma(\cdot,b)}{\vert\cdot\vert^2}\right\Vert_{L^\infty(\{\vert k\vert\leq \delta\})}\\
&\leq C_1(\beta) \Vert E\Vert_{L^1}+C_2(\beta)\Vert E\Vert_{L^2}\,, 
\end{split}
\]
by the Plancherel identity and the continuity of the Fourier transform from $L^1$ to $L^\infty$. Here, $C_1(\beta)$ and $C_2(\beta)$ are two positive constants depending on $\beta$. This concludes the proof of item $iv)$, thanks to the result in \cite{Mor-2026-AHP} which applies to the terms in \eqref{eq:F2} involving the multipliers $M^0(k)$ and $M^T(k,\beta)$, allowing the functional in \eqref{eq:FPV} to be extended to the space $X$ in \eqref{eq:X}.


\addtocontents{toc}{\protect\setcounter{tocdepth}{0}} 
\section*{Declarations and statements}

{\bf Research funding}. W. B. has been supported by MUR grant \emph{Dipartimento di Eccellenza 2023-2027} of Dipartimento di Matematica and by Gruppo Nazionale per l'Analisi Matematica, la Probabilità e le loro Applicazioni GNAMPA – INdAM in Italy. W.B. acknowledges that this study was carried out within the project E53D23005450006 ``Nonlinear dispersive equations in presence of singularities'' -- funded by European Union -- Next Generation EU within the PRIN 2022 program (D.D. 104 - 02/02/2022 Ministero dell'Universit\`a e della Ricerca). U. M. gratefully acknowledges financial support from MUR–Italian Ministry of University and Research and Next Generation EU within PRIN 2022AKRC5P `Interacting Quantum Systems: Topological Phenomena and Effective Theories'. U. M. has also been partially supported by Gruppo Nazionale per la Fisica Matematica GNFM – INdAM in Italy. U. M. acknowledges financial support from the European Union through the European Research Council’s Starting Grant \textsc{FermiMath}, grant agreement nr. 101040991. Views and opinions expressed are those of the authors and do not necessarily reflect those of the European Union or the European Research Council Executive Agency. Neither the European Union nor the granting authority can be held responsible for them.

{\bf Conflict of interest}. The authors declare that they have no competing interests regarding the publication of this paper. \\

{\bf Availability of data}. There are no data associated with the research in this paper.\\

{\bf Author contributions}. All authors contributed equally to the study, read and approved the final version of the submitted manuscript.

\addtocontents{toc}{\protect\setcounter{tocdepth}{2}} 


\end{document}